\title{Discerning Solution Concepts%
	\thanks{%
	    A previous version of this paper was circulated under the title ``Identification of Solution Concepts for Discrete Games''.
		We are thankful for the supervision of Joris Pinkse, Sung Jae Jun and Andr\'es Aradillas-L\'opez,
		as well as the useful comments from 
		Victor Aguiar,
		Roy Allen,
		Bulat Gafarov,
		Paul Grieco,
        Marc Henry,
        Brendan Kline,
		Robert Marshall,
		Francesca Molinari,
		Salvador Navarro, and
		Mark Roberts. 
		We also thank the attendants of the 2014 Spring Midwest Theory and Trade Conference at IUPUI,
		the 2014 Summer Meeting of The Econometric Society at the University of Minnesota,
		the 2015 International Game Theory Conference at Stony Brook University,
		and the 11th World Congress of the Econometric Society.
		We gratefully acknowledge the Human Capital Foundation (\href{http://www.hcfoundation.ru/en/}{http://hcfoundation.ru/en/}),
		and particularly Andrey P.\ Vavilov,
		for research support through the Center for the Study of Auctions, Procurement, and Competition Policy
		(\href{http://capcp.psu.edu/}{http://capcp.psu.edu/}) at the Pennsylvania State University.
		All remaining errors are our own.}
}
\author{
	Nail Kashaev%
	\thanks{
		Department of Economics, Western University,
		{\href{http://nail.kashaev.ru/}{{nail.kashaev.ru}}},
		{\href{mailto:nkashaev@uwo.ca}{{nkashaev@uwo.ca}}}.
	}
	\and
	Bruno Salcedo%
	\thanks{
		Department of Economics, Western University,
		{\href{http://www.brunosalcedo.com/}{{brunosalcedo.com}}},
		{\href{mailto:bsalcedo@uwo.ca}{{bsalcedo@uwo.ca}}}.
	}
}
\date{\enspace First version: January 3, 2014\\
    This version: \today}
\documentclass[12pt]{article}
\usepackage[english]{babel}    % Language: English
\usepackage[ansinew]{inputenc} % Input
\usepackage[T1]{fontenc}       % Font encoding
\usepackage{lmodern,microtype} % Typeface
\usepackage{amsmath,amsthm,amssymb,amsfonts} % AMS math
\usepackage{dsfont,mathrsfs,accents,bm} % Math style
\usepackage{titlesec,titling}  % Section titles
\usepackage[nohead]{geometry}  % Page & margins
\usepackage{setspace,caption}  % Spacing
\usepackage{enumitem,booktabs} % Tables & lists
\usepackage[comma]{natbib} % references
\usepackage{pstricks,pst-all}   % Figures
\usepackage{hyperref,breakurl}
\setenumerate{label=\small(\roman*)}
\hypersetup{breaklinks=true,colorlinks=true,pdfnewwindow=true,pdfstartview=FitH,%
	pdftitle="Discernibility",pdfauthor="Kashaev and Salcedo",%
	urlcolor=blue!95!red!50!black,citecolor=blue!95!red!50!black,linkcolor=red!90!black}
\DeclareCaptionFont{fancy}{\bfseries\sffamily} 
\newcommand{\ubar}[1]{\underaccent{\bar}{#1}}
\newcounter{aux}
\newcounter{egentry}
\allowdisplaybreaks

\newcommand{\zibri}{\"{U}///}
\newcommand{\quokkito}{\"{q}}
\AtEndDocument{\vfill\centering\footnotesize\textcolor{black!14}{\quokkito \quad \zibri }}
\providecommand{\psreset}{\psset{%
		linewidth=0.3pt,linestyle=solid,linecolor=black,
		dotsize=2.5pt,dotsep=2.5pt,arrowsize=4pt,
		fillstyle=none,fillcolor=white,
		showpoints=false,arrows=-,linearc=0,framearc=0,
		hatchsep=2pt,hatchwidth=0.2pt,nodesep=4pt,opacity=1}
	\psset{gridcolor=black!60, subgridcolor=black!30}
}

\psreset
\titleformat{\section}[block]{\centering\large\bfseries\sffamily}{\thesection.}{0.5em}{}
\titleformat{\subsection}[block]{\flushleft\bfseries}{\thesubsection.}{0.5em}{}
\titleformat{\subsection}[block]{\flushleft\bfseries\sffamily}{\thesubsection.}{0.5em}{}
\titleformat{\subsubsection}[runin]{\normalsize\itshape}{\bfseries\upshape\sffamily\thesubsubsection.}{0.5em}{}[.--\:]
\renewcommand{\thesubsubsection}{\arabic{section}.\arabic{subsection}.\alph{subsubsection}}
\titlespacing{\section}{0ex}{10ex}{5ex}
\titlespacing{\subsection}{0in}{6ex}{3ex}
\titlespacing{\subsubsection}{0mm}{2ex}{0.5em}
\pretitle{\vspace*{-10ex}\begin{center}\LARGE\bfseries\sffamily}
	\posttitle{\par\end{center}\vskip 0.5em}
\preauthor{\begin{center} \large \lineskip 0.5em\begin{tabular}[t]{c}}
		\postauthor{\end{tabular}\par\end{center}}
\predate{\begin{center}\small}
	\postdate{\par\end{center} \vspace*{-5ex} }
\providecommand{\abstitle}[1]{{\par\vspace*{2ex}\small\bfseries\sffamily #1}\hspace*{1ex}}
\renewenvironment{abstract}%
{\begin{center}\begin{minipage}{0.82\linewidth}%
			\setlength{\parindent}{0.0em}\abstitle{Abstract}\small}%
		{\end{minipage}\end{center}\vfill\clearpage}
\newtheoremstyle{defn}{3ex}{3ex}{}{}{\bfseries\sffamily}{}{.5em}%
{{\thmname{#1}\:\thmnumber{#2}}\:\thmnote{\mdseries({\small#3})\:}}%
\theoremstyle{defn}
\newtheorem{definition}{Definition}

\newtheorem{assumption}{Assumption}

\newtheoremstyle{prop}{3ex}{3ex}{\itshape}{}{\sffamily\bfseries}{}{.5em}%
{{\thmname{#1}\:\thmnumber{#2}}\:\thmnote{\mdseries({\small#3})\:}}
\theoremstyle{prop}
\newtheorem{proposition}{Proposition}[section]
\newtheorem{theorem}[proposition]{Theorem}
\newtheorem{lemma}[proposition]{Lemma}
\newtheorem{corollary}[proposition]{Corollary}
\newtheoremstyle{obs}{2ex}{2ex}{}{}{\itshape}{}{.5em}%
{{\thmname{#1}\:\thmnumber{#2}}\:\thmnote{[{\small#3}]\:}}
\theoremstyle{obs}
\newtheorem{example}{Example}%[section]
\providecommand{\blank}{{\,\cdot\,}}

\providecommand{\Exp}[2][\!]{\mathds{E}_{#1}\left[\,#2\,\right]}

\providecommand{\Conv}[1][]{\xrightarrow[#1]{\quad}}
\providecommand{\Char}{\mathds{1}}
\providecommand{\Real}{{\mathds{R}}}

\providecommand{\tr}{^{{\sf T}}}
\providecommand{\as}{\ensuremath{\mathrm{a.s.}}}
\providecommand{\rand}[1]{\mathbf{#1}}
\newcommand{\abs}[1]{\left\lvert#1\right\rvert}

\begin{document}
%%%
\maketitle
\begin{abstract}
	The empirical analysis of discrete complete-information games has relied on behavioral restrictions in the form of solution concepts, such as Nash equilibrium.
	%%%
    Choosing the right solution concept is crucial not just for identification of payoff parameters, but also for the validity and informativeness of counterfactual exercises and policy implications.
    %%%
    We say that a solution concept is discernible if it is possible to determine whether it generated the observed data on the players' behavior and covariates. 
    %%%
    We propose a set of conditions that make it possible to discern solution concepts.
    %%%
    In particular, our conditions are sufficient to tell whether the players' choices emerged from Nash equilibria.
    %%%
    We can also discern between rationalizable behavior, maxmin behavior, and collusive behavior.
    %%%
    Finally, we identify the correlation structure of unobserved shocks in our model using a novel approach.
	%%%
	\abstitle{Keywords} 
	Discrete Games
	$\cdot$ Testability 
	$\cdot$ Identification
	$\cdot$ Incomplete models
	$\cdot$ Market entry
	%%%
	\abstitle{JEL classification} 
	C52 $\cdot$ C72
\end{abstract}
	
%%%%%%%%%%%%%%%%%%%%%%%%%%%%%%%%%%%%%%%%%%%%%%%%%%%%%%%%%%%%%%%%%%%%%%%
%%%%%%%%%%%%%%%%%%%%%%%%%%%%%%%%%%%%%%%%%%%%%%%%%%%%%%%%%%%%%%%%%%%%%%%
%%%%%%%%%%%%%%%%%%%%%%%%%%%%%%%%%%%%%%%%%%%%%%%%%%%%%%%%%%%%%%%%%%%%%%%
%%%%%%%%%%%%%%%%%%%%%%%%%%%%%%%%%%%%%%%%%%%%%%%%%%%%%%%%%%%%%%%%%%%%%%%
\section{Introduction}\label{sec:intro}
%%%
In Game Theory, solution concepts impose restrictions on the behavior of players given their payoffs. 
%%%
The most popular solution concept is Nash equilibrium (NE) \citep{nash51}. 
%%%
Solution concepts are often used to establish theoretical results, to identify payoff parameters, and to derive policy and welfare implications from counterfactual analyses.%
%%%
\footnote{%
	The classic revealed-preference approach to the identification of payoff parameters in discrete games of complete information assumes that the choice of each player is a best response to the observed choices of other players \citep{bjorn84,jovanovic89,bresnahan90}.
	%%%
	This is tantamount to assuming that the players' choices constitute pure strategy NE. 
	%%%
	The approach can be generalized to allow mixed strategy NE \citep{tamer03,BHR},
		rationalizable strategies \citep{aradillas08,kline15},
		or general convex solution concepts \citep{BMM,galichon11}.
	%%%
	See \cite{dePaula13} for a review of the literature.} 
%%%
However, there may exist different solution concepts that are observationally equivalent but yield different payoff parameters, theoretical implications, and counterfactual predictions (see Section~\ref{sec:eg} for an example).
%%%
Thus, it is important to understand when one can tell a solution concept apart from other solution concepts.
%%%
In such cases, we sat that the solution concept is \emph{discernible}.

We consider multiplayer binary-action games of complete information, similar to the classic entry game from \cite{bresnahan90}.
%%%
We maintain the assumption that the players' choices can display any form of rationalizable behavior in the sense of \citet{bernheim1984} and \citet{pearce84}.
%%%
We provide a set of conditions that are sufficient to establish discernibility of any solution concept stronger than rationalizability.
%%%
For instance, it is possible to determine whether the players' decisions arise from NE.
%%%
Moreover, if they do arise from NE, then they cannot be consistent with any other form of rationalizable behavior.
%%%
We also identify all the payoff parameters, including those governing the correlation structure of the unobserved heterogeneity.
%%%
To the best of our knowledge, this is the first formal result that identifies the correlation parameters using a solution concept weaker than pure-strategy NE. 

Usually, the testable implications of a solution concept (e.g., NE) can be used to determine whether it is consistent with or \emph{could have} generated the observed data.%
%%%
\footnote{For example, one can construct a models specification test based on the results from \cite{BMM} or \cite{galichon11}.}
%%%
Our results allow the researcher to answer the question of whether the solution concept actually generated the data.
%%%
This question is important because of several reasons. 
%%%
A solution concept can be consistent with the data and, at the same time, yield misleading counterfactual predictions. 
%%%
For example, this could happen if an alternative solution concept generated the data, and the two solution concepts are observationally equivalent. 
%%%
We provide an example in Section~\ref{sec:eg}.
%%%
Establishing discernibility of NE precludes this possibility and helps to establish the validity of counterfactual analysis and policy implications.

Discernibility is also useful in making sharper counterfactual predictions. 
%%%
One can always assume a less restrictive solution concept (in our case rationalizabilty), build the confidence set for the payoff parameters, and then construct robust confidence bands for the counterfactual of interest.
%%%
However, these bands can be uninformative because of the weakness of the restrictions imposed on behavior.
%%%
If one shows that a stronger solution concept (e.g., NE) generated the data, and this solution concept is discernible, then one can build more informative bounds for the counterfactual predictions.
%%%
In other words, our methodology allows the researcher to determine the strongest restrictions on behavior that are still consistent with the observed data.

Discernibility also has practical implications that may reduce the computational burden. 
%%%
Suppose that the researcher is considering different solution concepts,
and establishes discernibility of all of them.
%%%
Discernibility implies that at most one of the solution concepts under consideration can explain the data.
%%%
Hence, if a given solution concept explains the data, then the researcher can automatically rule out the other alternatives. 

Our strategy to establish discernibility relies on two assumptions. 
%%%
First, we assume that the researcher observes covariates with full support satisfying an exclusion restriction. 
%%%
Second, we assume that the excluded covariates generate enough variation in the conditional distribution of payoffs conditional on covariates.
%%%
In particular, we require the family of these conditional distributions to be boundedly complete.%
%%%
\footnote{%
	Completeness of a family distribution is a well-known concept both in Statistical and Econometrics literature. See \cite{andrews11}.
	\cite{newey03} and \cite{darolles11} use a completeness assumption to establish non-parametric identification for conditional moment restrictions. \cite{blundell07} use it to achieve identification of Engel curves. \cite{hoderlein12} impose bounded completeness in the context of structural models with random coefficients.}
%%%
Using these assumptions, we identify the distribution of payoffs and the distribution of outcomes conditional on both the observed and unobserved characteristics of the environment.
%%%
Knowing these distributions allows us to establish discernibility of solution concepts. 

We are not the first to exploit the power of completeness assumptions coupled with exclusion restrictions to discriminate between behavior patterns.
%%%
\cite{berry14} apply a strategy similar to ours to a model of oligopolistic competition that allows, among other things, to discriminate between different models of competition.
	%%%
	A significant difference between their setting and ours is that they consider continuous games, while we consider discrete games.
	%%%
They crucially rely on having an uncountable set of outcomes to relax the completeness assumption to some extent.

Identification of the payoff parameters is not necessary for discernibility. 
In Section~\ref{sec:beyond}, we relax rationalizability and allow for some forms of collusive behavior and ambiguity aversion in the sense of \cite{gilboa89}. 
%%%
This comes at the expense that some of the payoff parameters are no longer point identified.
%%%
However, we still can establish discernibility of a large class of solution concepts. 

%%%%%%%%%%%%%%%%%%%%%%%%%%%%%%%%%%%%%%%%%%%%%%%%%%%%%%%%%%%%%%%%%%%%%%%
%%%%%%%%%%%%%%%%%%%%%%%%%%%%%%%%%%%%%%%%%%%%%%%%%%%%%%%%%%%%%%%%%%%%%%%
%%%%%%%%%%%%%%%%%%%%%%%%%%%%%%%%%%%%%%%%%%%%%%%%%%%%%%%%%%%%%%%%%%%%%%%
%%%%%%%%%%%%%%%%%%%%%%%%%%%%%%%%%%%%%%%%%%%%%%%%%%%%%%%%%%%%%%%%%%%%%%%
\section{Motivating Example}
\label{sec:eg}
%%%
We begin with a simple example to motivate the meaning and the importance of discernibility of solution concepts.
%%%
First, we show that two different solution concepts%
---pure strategy Nash equilibrium (PNE) and a behavioral solution concept called strategic ambiguity aversion (SAA)---can be observationally equivalent.
%%%
That is, they can generate the same distribution over observables.
%%%
Hence, it is impossible to discern PNE and SAA in our example.
%%%
Next, we show that PNE and SAA not being discernible can lead to incorrect quantitative and qualitative policy recommendations. 

%Consider an entry model similar to the classic entry game from \cite{bresnahan90}.
%%%
Two firms $i\in \{1,2\}$ simultaneously choose whether to enter a market $(y_i=1)$ or not $(y_i=0)$.
%%%
Firm $i$'s profit is given by
\[
%\label{eqn:profits-eg}
	y_{i} \cdot\big[\eta_0 (1- y_{-i}) - \rand{e}_{i} \big],%
%%%
\footnote{Throughout the paper, we use boldface font to denote random variables and vectors, and regular font for deterministic ones.}
\]
%%%
where 
    ({i}) $y_{-i}$ is the choice of $i$'s competitor;
    %%%
    ({ii}) $\eta_0 \ge0$ is a fixed parameter that measures the effect of competition and is unknown by the researcher;
    %%%
    and ({iii}) $\rand{e} = (\rand{e}_1,\rand{e}_2)$ is a vector of payoff shocks unobserved by the researcher.
    %%%
    We assume that $\rand{e}$ is supported on $\Real^2$ and admits a probability density function that is symmetric around the $45$-degree line (e.g., $\rand{e}_1$ and $\rand{e}_2$ are independent standard normal random variables).
    % and ({iii}) $\rand{e} = (\rand{e}_1,\rand{e}_2)$ is a vector of payoff shocks unobserved by the researcher and normally distributed with zero mean and the identity matrix as a covariance matrix.
%%%
The firms observe both $\eta_0$ and $\rand{e}$.
%%%
That is, the game is of complete information.
%%%
The researcher observes (can consistently estimate) the distribution of outcomes $\rand{y} = (\rand{y}_1,\rand{y}_2)$.
%%%
Many of the simplifications we make in this example are for exposition purpose only and are relaxed in subsequent sections.

%%%%%%%%%%%%%%%%%%%%%%%%%%%%%%%%%%%%%%%%%%%%%%%%%%%%%%%%%%%%%%%%%%%%%%%%%%%%%%%%%%%%%
%%%%%%%%%%%%%%%%%%%%%%%%%%%%%%%%%%%%%%%%%%%%%%%%%%%%%%%%%%%%%%%%%%%%%%%%%%%%%%%%%%%%%
%%%%%%%%%%%%%%%%%%%%%%%%%%%%%%%%%%%%%%%%%%%%%%%%%%%%%%%%%%%%%%%%%%%%%%%%%%%%%%%%%%%%%
\subsection{Failure of Discernibility}
\label{sec:eg-discernibility}
%%%
Following \cite{bresnahan90},
the classic approach to analyze entry games is to assume that the firms' choices always constitute PNE. 
%%%
Suppose that a researcher wants to test this assumption under the milder assumption that firm behavior is rationalizable. 
%%%
In our example, rationalizability is equivalent to assuming that the firms can choose any action that survives two rounds of elimination of strictly dominated strategies.
%%%
When $e_i < 0$, entering the market is strictly dominant for firm $i$.
%%%
When $e_i > \eta_0$, staying out of the market is strictly dominant for firm $i$.
%%%
This results in four regions of the payoff space in which the game has a unique rationalizable outcome.
%%%
In the remaining region---the \emph{multiplicity region}---rationalizability imposes no restrictions on behavior.

Consider the following solution concept. 
%%%
Firm behavior is rationalizable, but firms never enter when there are multiple rationalizable outcomes. 
%%%
This could happen, for instance, if the firms were ambiguity averse and used maxmin strategies when facing strategic uncertainty.
%when the game is not dominance solvable. 
%%%
This solution concept is called SAA and is analyzed in \cite{mass19}.
%%%
The predictions of SAA with $\eta_0=\eta$ for different realizations of $\rand{e}$ are illustrated in the left panel of Figure~\ref{fig:eg}.

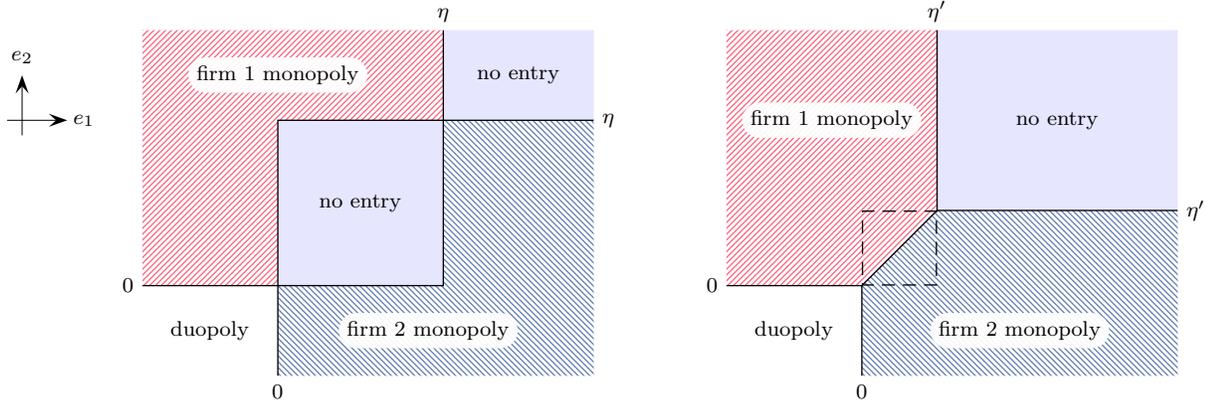
\begin{figure}[t]
	\centering
	\psset{unit=4mm}
	\newrgbcolor{navy}{0.1 0.3 0.5}
	\begin{pspicture}(-6,-4)(11,8.5)
	%\psgrid
	\scriptsize%
	\newrgbcolor{uno}{0.4 0.5 0.7}
	\newrgbcolor{dos}{0.99 0.40 0.50}
	% Axes
 	\psset{arrows=->,arrowsize=5pt}
 	\psline(-8.5,5)(-8.5,7) \rput[b](-8.5,7.4){$e_2$}
 	\psline(-9,5.5)(-7,5.5) \rput[l](-6.8,5.5){$e_1$} 
 	\psreset%
	% objects
	%\psset{hatchsep=1.75pt,hatchwidth=0.3pt}
	\psset{hatchsep=1pt,hatchwidth=0.5pt}
	% regions
	\psframe*[linecolor=blue!10](5.5,5.5)(10.5,8.5)
	\psframe*[linecolor=blue!10](0,0)(5.5,5.5)
	\pspolygon[linestyle=none,fillstyle=vlines,hatchcolor=uno](0,-3)(0,0)(5.5,0)(5.5,5.5)(10.5,5.5)(10.5,-3)
	\pspolygon[linestyle=none,fillstyle=hlines,hatchcolor=dos](-4.5,0)(0,0)(0,5.5)(5.5,5.5)(5.5,8.5)(-4.5,8.5)
	% delimiting lines
	\psset{linewidth=0.5pt}
	\psline(-4.5,0)(5.5,0)(5.5,8.5)
	\psline(0,-3)(0,5.5)(10.5,5.5)
	\psreset%
	% Labels
	\psset{linestyle=none,fillstyle=solid,opacity=0.95,framearc=1}
	% thresholds
	\rput[l](10.5,5.5){{\psframebox{$\eta$}}}
	\rput[r](-4.5,0){{\psframebox{$0$}}}
	\rput[b](5.5,8.5){\psframebox{$\eta$}}
	\rput[t](0,-3.0){{\psframebox{$0$}}}
	% regions
	\rput(8,7){{no entry}}
	\rput(2.75,2.75){{no entry}}
	\rput(0,7.0){\psframebox{firm $1$ monopoly}}
	\rput(5,-1.5){\psframebox{firm $2$ monopoly}}
	\rput(-2.25,-1.5){\psframebox{duopoly}}
	\psreset%
	\end{pspicture}
	\qquad
	%%% Figure: Rationalizing ambiguity with Nash
	\begin{pspicture}(-6,-4)(11,8.5)
	%\psgrid
	\scriptsize%
	\newrgbcolor{uno}{0.4 0.5 0.7}
	\newrgbcolor{dos}{0.99 0.40 0.50}
	% objects
	\psset{hatchsep=1pt,hatchwidth=0.5pt}
	% regions
	\psframe*[linecolor=blue!10](2.5,2.5)(10.5,8.5)
	\pspolygon[linestyle=none,fillstyle=vlines,hatchcolor=uno](0,-3)(0,0)(2.5,2.5)(10.5,2.5)(10.5,-3)
	\pspolygon[linestyle=none,fillstyle=hlines,hatchcolor=dos](-4.5,0)(0,0)(2.5,2.5)(2.5,8.5)(-4.5,8.5)
	%\psframe[linestyle=none,fillstyle=vlines,hatchcolor=red!60!blue](-0.5,-2)(4,1.5)
	% delimiting lines
	\psset{linewidth=0.5pt}
	\psline(-4.5,0)(0,0)(0,-3)
	\psline(2.5,8.5)(2.5,2.5)(10.5,2.5)
	\psline(0,0)(2.5,2.5)
	\psset{linestyle=dashed}
		\psframe(0,0)(2.5,2.5)
	\psreset%
	% Labels
	\psset{linestyle=none,fillstyle=solid,opacity=0.95,framearc=1}
	% thresholds
	%\rput[l](9,3){{\psframebox{$z_2+\eta$}}}
	\rput[l](10.5,2.5){{\psframebox{$\eta'$}}}
	\rput[r](-4.5,0){{\psframebox{$0$}}}
	\rput[b](2.5,8.5){\psframebox{$\eta'$}}
	\rput[t](0,-3){{\psframebox{$0$}}}
	% regions
	\rput(6.5,5.5){{no entry}}
	\rput(-1,5.5){\psframebox{firm $1$ monopoly}}
	\rput(5.25,-1.5){\psframebox{firm $2$ monopoly}}
	\rput(-2.25,-1.5){\psframebox{duopoly}}
	%\rput(1.75,-0.25){\psframebox[framesep=2pt]{$\begin{array}{\text{random} monopoly}}
	\psreset%
	\end{pspicture}
	\caption{Two observationally equivalent models:
		strategic ambiguity aversion (left)
		and pure-strategy Nash equilibrium (right).}
	\label{fig:eg}
\end{figure}

We will show that PNE and SAA can produce the exact same distributions over observables.
%%%
According to PNE, only one firm enters in the multiplicity region.
%%%
Since either of the two firms could be the one that enters in equilibrium, we need to specify an equilibrium selection rule. 
%%%
We assume that firms always play the equilibria according to which the most profitable firm is the one that enters.
%%%
The predictions of PNE with such selection rule and $\eta_0=\eta'$ are illustrated in the right panel of Figure~\ref{fig:eg}.

The duopoly region is the same under both solution concepts.
%%%
For a fixed value of the competition effect, the no-entry region is smaller under PNE.
%%%
However, since $\eta_0$ is unknown to the researcher, it is possible to set $\eta'<\eta$ so that both models assign the same probability to no entry.
%%%
Since both models imply the same probability of both duopoly and no entry, 
they also imply the same probability of having a monopoly. 
%%%
Note that both the distribution of shocks and each of the two models are symmetric around the 45-degree line.
%%%
Hence, each of the two monopolies are equally likely under both solution concepts. 
%%%
The formal proof is in Appendix~\ref{sec:eg-proof} in the online supplement.

Despite being very different, SAA and PNE can imply identical distributions over outcomes. 
%%%
Hence, in this example, it is impossible to determine whether the data \emph{is} generated by PNE. 
%%%
At best, the researcher can tell whether the data \emph{can be} explained by PNE.
%%%
That is, PNE is not discernible. 
% %%%
Exactly for the same reason, SAA is also not discernible. 
%%%
The fact that we use different parameter values for each of the two models is unimportant.
%%%
For instance, Proposition~\ref{prop:nash} in Section~\ref{sec:main-discernibility} establishes a general nondiscernibility result that can be applied even if the payoff parameters are assumed to be the same under different solution concepts.

%%%%%%%%%%%%%%%%%%%%%%%%%%%%%%%%%%%%%%%%%%%%%%%%%%%%%%%%%%%%%%%%%%%%%%%%%%%%%%%%%%%%%
%%%%%%%%%%%%%%%%%%%%%%%%%%%%%%%%%%%%%%%%%%%%%%%%%%%%%%%%%%%%%%%%%%%%%%%%%%%%%%%%%%%%%
%%%%%%%%%%%%%%%%%%%%%%%%%%%%%%%%%%%%%%%%%%%%%%%%%%%%%%%%%%%%%%%%%%%%%%%%%%%%%%%%%%%%%
\subsection{Importance of Discernibility}
\label{sec:eg-policy}
%%%%
Next, we show that PNE and SAA can generate opposite counterfactual predictions in our example.
%%%
Therefore, the failure of discernibility can lead to incorrect policy recommendations.
%%%
Suppose that a policymaker wants to increase the number of markets that are served by at least one firm. 
%%%
As a policy instrument, she can choose to offer a subsidy $\tau>0$ to one firm, say Firm $1$, for entering markets in which Firm $2$ does not enter.
%%%
Entry subsidies are commonly used to incentivize the provision of strategic infrastructure such as broadband internet access \citep{goolsbee02}.
%%%
The specific subsidy scheme we analyze allows for a stark and simple exposition.
%%%
Appendix~\ref{sec:subsidy} in the online supplement presents a similar result with a more realistic subsidy scheme.

Suppose that the data is generated by SAA, but the policymaker evaluates the policy assuming that firms always play PNE.
%%%%
We have already established that the policymaker cannot refute her assumption because PNE and SAA are observationally equivalent.
%%%
However, under PNE the size of the competition effect must be smaller than the one under SAA (i.e., $\eta'<\eta$).
%%%
Thus, assuming the incorrect solution concept would lead to inconsistent estimates of $\eta_0$.
%%%
This, in turn, might lead to flawed welfare evaluations. 
%%%
Since both models are examples of rationalizable behavior,
one may expect that an inconsistent estimator of the competition effect will not affect the qualitative implications of different policy interventions. 
%%%
Here, this is not the case. 

%%%%%%%%%%%%%%%%%%%%%%%%%%%%%%%%%%%%%%%%%%%%%%%%%%%%%%%%%%%%%%%%%%%%%%%%%%%%%%%%%%%%%%%%%%%%%%%%%%%%
PNE in the multiplicity region always predict monopolies.
%%%
Hence, under the PNE hypothesis, all markets are served except for those in which not entering is dominant for both firms.
%%%
The policy being evaluated decreases the probability of the latter region (see Figure~\ref{fig:eg-policy}).
%%%
Therefore, under the policymaker's assumptions, the policy unambiguously reduces the number of markets without service independently of the parameter values.

\begin{figure}[t]
	\centering
	\psset{unit=4mm}
	\newrgbcolor{navy}{0.1 0.3 0.5}
	\begin{pspicture}(-11,-5)(20,8)
	%\psgrid
	\scriptsize%
	% Axes
% 	\psset{arrows=->,arrowsize=5pt}
% 	\psline(-8.5,5)(-8.5,7) \rput[b](-8.5,7.4){$e_2$}
% 	\psline(-9,5.5)(-7,5.5) \rput[l](-6.8,5.5){$e_1$} 
% 	\psreset%
	% regions
	\newrgbcolor{gain}{0.99 0.40 0.50}
	\newrgbcolor{loss}{0.4 0.5 0.7}
	\psset{hatchsep=1pt,hatchwidth=0.5pt,linestyle=none}
	% gains
	\psset{fillstyle=hlines,hatchcolor=gain}
	\psframe(3.5,3)(7,7)
	\psframe(11,-1)(13,1)
	\rput[l](13.5,0){\parbox{28mm}{\flushleft Shrinkage of no-entry region}}
	% losses
	%\psset{fillstyle=vlines,hatchcolor=loss}
	\psset{fillstyle=solid,fillcolor=blue!10}
	\psframe(3.5,-1)(7,3)
	\psframe(11,-4)(13,-2)
	\rput[l](13.5,-2.5){ \parbox{28mm}{\flushleft Growth of multiplicity region}}
	\psreset%
	% delimiting lines
	\psline(-2,-4)(-2,3)(11,3)
	\psline(-6,-1)(7,-1)(7,7)
	\psline(-6,-1)(3.5,-1)(3.5,7)
	% policy arrows
	\newrgbcolor{policy}{0.2 0.5 0.3}
	\psset{arrows=->,linecolor=policy,linewidth=1.5pt, arrowsize= 5pt}
	\psline(3.7,7.6)(6.8,7.6)
	\psreset%
	% labels
	\psset{linestyle=none,fillstyle=solid,opacity=0.95,framearc=1}
	% old thresholds
	\rput[l](11,3){{\psframebox{$\eta$}}}
	\rput[r](-6,-1){{\psframebox{$0$}}}
	\rput[br](3.5,7){\psframebox{$\eta$}}
	\rput[t](-2,-4){{\psframebox{$0$}}}
	% new threshholds
	\rput[bl](7,7){{\psframebox{$\eta + \tau$}}}
	% areas
	\rput(5.25,5){\psframebox{$E^+$}}
	\rput(5.25,1){{$E^-$}}			
	\psreset%
	\end{pspicture}
	\caption{Effect of the proposed policy.}
	\label{fig:eg-policy}
\end{figure}
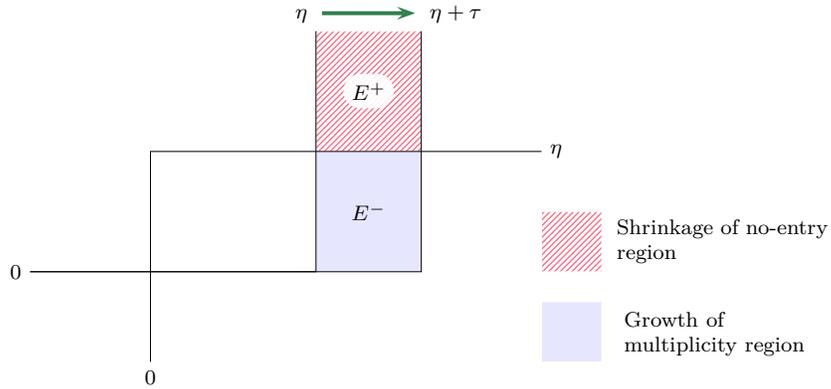

However, under SAA, the effect of the policy is always smaller than under PNE, and it can even have the opposite direction for some parameter values. 
%%%
This can happen because the policy also increases the probability of the multiplicity region and, under strategic ambiguity aversion, firms never enter in this region.
%%%
A firm might be willing to forego the subsidy for fear of another firm entering the market, which would result in negative profits. 
%%%
The net effect of the policy on the probability of monopolies is given by the difference between the probabilities of regions $E^+$ and $E^-$ in Figure~\ref{fig:eg-policy}. 
%%%
The adverse effect can actually dominate and the policy can increase the probability that a market is not served.
%%%
For example, one can verify that this is the case whenever $\rand{e}_1$ and $\rand{e}_2$ are independent standard normal random variables and $\Phi(\eta)>3/4$,
where $\Phi(\blank)$ is the standard normal cumulative distribution function.
%%%%%%%%%%%%%%%%%%%%%%%%%%%%%%%%%%%%%%%%%%%%%%%%%%%%%%%%%%%%%%%%%%%%%%%
%%%%%%%%%%%%%%%%%%%%%%%%%%%%%%%%%%%%%%%%%%%%%%%%%%%%%%%%%%%%%%%%%%%%%%%
%%%%%%%%%%%%%%%%%%%%%%%%%%%%%%%%%%%%%%%%%%%%%%%%%%%%%%%%%%%%%%%%%%%%%%%
%%%%%%%%%%%%%%%%%%%%%%%%%%%%%%%%%%%%%%%%%%%%%%%%%%%%%%%%%%%%%%%%%%%%%%%
\section{Framework}
\label{sec:model}

The motivating example from Secion \ref{sec:eg} shows that it is possible for very different solution concepts to be observationally equivalent while implying different policy recommendations. 
%%%
In what follows, we introduce a framework that rules out that possibility. 
%%%
Our assumptions on the data generating process and the covariates observed by the researcher guarantee that a large class of solution concepts are discernible, in a formal sense to be defined. 

%%%%%%%%%%%%%%%%%%%%%%%%%%%%%%%%%%%%%%%%%%%%%%%%%%%%%%%%%%%%%%%%%%%%%%%%%%%%%%%%%%%%%%%%%%%%%%%
\subsection{Payoffs}
\label{sec:model-u}
%%%
There are $d_I<+\infty$ players (firms) indexed by $i \in I = \{1,\ldots, d_I\}$.
%%%
Each player chooses an action $y_i\in Y_i = \{0,1\}$.
%%%
Appendix~\ref{sec:app-multi} presents a generalization to games with many actions. 
%%%
The set of outcomes is $Y = \times_{i\in I} Y_i$. 
%%%
Let $y_{-i}=(y_j)_{j\neq i}$ denote the vector of actions from $i$'s opponents.
%%%
Player $i$'s payoffs from outcome $y$ are given by
%%%
\[
y_{i} \cdot\big(
\alpha^{0}_{i,y_{-i}}(\rand{w}) + \beta^{0}_i(\rand{w}) \rand{z}_i - \rand{e}_i
\big),%
%%%
\footnote{We use boldface font to denote random variables and vectors.}
\]
where
({i}) $\rand{w}$ is a vector of observed player and market characteristics with support $W\subseteq \Real^{d_W}$;
({ii}) $\rand{z} = (\rand{z}_{i})_{i\in I}$ is a vector of player-specific covariates;
({iii}) $\rand{e} = (\rand{e}_{i})_{i\in I}$ is a vector of payoff shocks unobserved by the researcher;
and
({iv}) $\beta_i^{0},\alpha^{0}_{i,y_{-i}}:W\to\Real$ are unknown functions.
We assume that all the parameters and payoff shocks are common knowledge among the players. 
%%%
That is, the game is of complete information.

\setcounter{egentry}{\value{example}}
\begin{example}\label{eg:entry}
	If $\alpha^{0}_{i,y_{-i}}(w) = \delta_{ii}^0(w) + \sum_{j\neq i} \delta_{ij}^{0}(w) y_j$,
	then the specification corresponds to an entry game as in \cite{bresnahan90} or \cite{berry92}. 
	%%%
	The value of $\delta^{0}_{ii}(w)+\beta^{0}_i(w)z_i-e_i$ measures the monopoly payoff of firm $i$. 
	%%%
	For $i\neq j$ the value of $\delta_{ij}^{0}(w)$ measures the strategic effect of the market presence of firm $j$ on $i$'s profits.
	%%%
	The strategic effect of the presence of a firm on payoffs of competitors (the signs of $\delta_{ij}^{0}(w)$, $i\neq j$) is unrestricted. 
\end{example}
\begin{example}\label{eg:coordination}
	If $\alpha^{0}_{i,y_{-i}}(w) = \delta_i^{0}(w)\cdot \Char(\sum_{j\in I}y_j \geq \tau^{0}(w))$,
	with $\tau^{0}(w) \in [0,d_I]$,
	the model corresponds to a regime-change game. 
	%%%
	The term $\beta^{0}_i(w) z_i - e_i$ captures
	the individual cost of participating in a revolt. 
	%%%
	The threshold $\tau^{0}(w)$ determines the number of participants required for the revolt to 
	be successful. 
	%%%
	And $\delta_i^{0}(w)$ captures the benefit to $i$ from participating in a successful revolt. 
	%%%
	This payoff structure can also be used to study coordinated action problems \citep{rubinstein89}, bank-runs and currency attacks \citep{morris03}, 
	and tacit collusion in oligopolistic markets \citep{green13}.
\end{example}

We impose the following standard assumptions on payoffs (see, for instance, \citet{jia2008happens}, \citet{ciliberto09}, \citet{BHR}, and \citet{ciliberto2018market}).
	
\begin{assumption}\label{ass:z}{}\ 
	\begin{enumerate}
    \item The support of $\rand{z}$ conditional on $\rand{w}=w$ is $Z=\Real^{d_I}$ for all $w\in W$.
	\item $\beta_{i}^{0}(w)\neq 0$ for all $i$ and $w \in W$.
	\end{enumerate}
\end{assumption}
Assumption~\ref{ass:z} requires the player-specific covariates to have full support and be relevant.
%%%
In entry games, examples of continuous firm-specific covariates could be the logarithm of the distance of the market to the existing network of each firm, or to the firms' headquarters. 
%%%
These distances have been used by \cite{ciliberto09} and \cite{ciliberto2018market} to analyze the airline industry.
%%%
While there is a restriction on $\beta_i^{0}(\blank)$, we do not impose any restrictions on $\alpha^{0}_{i,-y}(\blank)$.

\begin{assumption}\label{ass:e}{}\ 
$\rand{e}|(\rand{z}=z,\rand{w}=w) \sim N(0,\Sigma^{0}(w))$ for all $z$ and $w$, where $\Sigma^{0}:W\to\Real^{d_I\times d_I}$ is such that $\Sigma^{0}(w)$ is a positive definite symmetric matrix and $\Sigma_{ii}^{0}(w)=1$ for all $w\in W$ and $i\in I$.
\end{assumption}

The normality assumption is common in applied work and helps to simplify the exposition.
%%%
In Appendix~\ref{sec:binary}, we replace it with two weaker assumptions.
%%%
The first one imposes restrictions on the tails of the distribution of $\rand{e}$.
The second one requires the distribution of $\rand{e}$ to constitute a boundedly complete family of distributions.
%%%
The normality assumption also implies that the probability that a player obtains the same payoffs from different outcomes is zero.
%%%
Thus, we do need to worry about situations when players may be indifferent between actions. 
%%%
The requirement $\Sigma^{0}_{ii}(w)=1$ is a scale normalization. 
%%%
Note that we allow the payoff shocks to be correlated across players.

To keep the notation tractable, we group covariates and payoff parameters as follows.
%%%
Let $\rand{x}=(\rand{z},\rand{w})$ be the vector of all observed covariates.
%%%
Let $\alpha=(\alpha_{i,y_{-i}}(\blank))_{i\in Y,y\in Y}$, $\beta=(\beta_i(\blank))_{i\in I}$, and $\theta=(\alpha,\beta,\Sigma(\blank))\in\Theta$.
%%%
Hence, we can define the \emph{payoff indices} $\pi(x,e,\theta)=(\pi_{i,y}(x,e,\theta))_{i\in I,y\in Y}$ by
\[
\pi_{i,y}(x,e,\theta) 
= y_i \cdot \left(\alpha_{i,y_{-i}}({w}) + \beta_{i}({w}) {z}_{i} - {e}_{i}\right).
\]
The true value of the payoff parameters is denoted by $\theta_0=(\alpha_0,\beta_0,\Sigma^{0}(\blank))$.

%%%%%%%%%%%%%%%%%%%%%%%%%%%%%%%%%%%%%%%%%%%%%%%%%%%%%%%%%%%%%%%%%%%%%%%%%%%%%%%%%%%%%%%%%%%%%%%
%%%%%%%%%%%%%%%%%%%%%%%%%%%%%%%%%%%%%%%%%%%%%%%%%%%%%%%%%%%%%%%%%%%%%%%%%%%%%%%%%%%%%%%%%%%%%%%
\subsection{Distribution of Play}
\label{sec:model-dop}
%%%
An important object for our analysis is the \emph{distribution of play} $h_{0}$,
defined as the conditional distribution of $\rand{y}$ given $\rand{x}$ and $\rand{e}$.
%%%
That is, 
\[
h_{0}(y,x,e) = \Pr(\rand{y} = y|\rand{x}=x,\rand{e}=e).
\]
%%%
The distribution of play describes the joint behavior of the players as a function of market and player characteristics. 
%%%
It is a nonparametric latent parameter. 
%%%
Let $h_{0}(x,e)=(h_{0}(y,x,e))_{y\in Y}$,
and let $H$ be the set of all possible distributions of play.
%%%
Given $h,h'\in H$, we say that $h=h'$ if and only if $h(\rand{x},\rand{e})=h'(\rand{x},\rand{e})\ \as$.
%%%
Note that, by construction, $h_{0}(x, e)$ belongs to the $d_Y$-dimensional simplex.

We impose the following restriction on the distribution of play. 
%%%
It limits the way the player-specific covariates and shocks 
can affect the behavior of the players.

\begin{assumption}[Exclusion Restriction]\label{ass:ER}
	There exists a measurable function $\tilde{h}_0$ such that
	$h_{0}(\rand{x},\rand{e})=\tilde{h}_0(\rand{w},\rand{v})\ \as$,
	where $\rand{v} = (\rand{v}_{i})_{i\in I}$ and $\rand{v}_{i} = \beta_{i}^{0}(\rand{w})\rand{z}_{i} - \rand{e}_{i}$.
\end{assumption}
Assumption~\ref{ass:ER} is a joint assumption on $h_{0}$ and $\theta_{0}$.
%%%
It requires that the player-specific covariates and shocks can affect choices \emph{only} via the index $\rand{v}$, whose distribution depends on the value of $\theta_{0}$.
%%%
It says that $\rand{y}$ is independent of $\rand{z}$ and $\rand{e}$ conditional on $\rand{v}$ and $\rand{w}$. 
%%%
Note that, given $\theta_{0}$, the realizations of $\rand{v}$ and $\rand{w}$ are sufficient to pin down the payoff indices.
%%%
Hence, Assumption~\ref{ass:ER} can be interpreted as requiring that,
after conditioning on the realization of $\rand{w}$, the players are payoff driven.
%%%
If there are two markets with the exact same payoff indices and the same realization of $\rand{w}$, then the distribution over outcomes should be the same. 
%%%
Assumption~\ref{ass:ER} is implied by the assumptions made in \cite{BHR}.%
%%%
\footnote{%
	More specifically, \cite{BHR} assume that players make choices 
	randomizing among the different NE of the game.
	Their Assumption 6 requires the selection probabilities
	to be measurable with respect to the latent utility indices, 
	in our notation.
	%%%
	An analogous assumption could be imposed on the selection mechanisms 
	of any model satisfying Assumptions 2.2--2.4 in \cite{BMM}.
	%%%
	Doing so would imply our Assumption~\ref{ass:ER}.}

%%%
Under Assumption~\ref{ass:ER} the distribution of play can be robust to policy interventions.
%%%
For example, if one wants to evaluate policies that only affect firms indirectly through the prices of inputs. 
%%%
The firms might care about the changes in prices, but not about the source of these changes. 
%%%
In situations where this assumption is reasonable, knowing $\theta_0$ and $h_0$ is sufficient to analyze policies that only operate through the payoff indices. 

%%%%%%%%%%%%%%%%%%%%%%%%%%%%%%%%%%%%%%%%%%%%%%%%%%%%%%%%%%%%%%%%%%%%%%%%%%%%%%%%%%%%%%%%%%%%%%%
%%%%%%%%%%%%%%%%%%%%%%%%%%%%%%%%%%%%%%%%%%%%%%%%%%%%%%%%%%%%%%%%%%%%%%%%%%%%%%%%%%%%%%%%%%%%%%%
\subsection{Solution Concepts}
\label{sec:model-sol}
%%%
Although Assumption~\ref{ass:ER} imposes some structural restrictions to the behavior of players, one may still want to impose additional economic restrictions.
%%%
These restrictions come in the form of solution concepts such as rationalizability or NE.
%%%
Solution concepts often depend on the characteristics of the environment. 
%%%
Hence, we allow for the restrictions arising from solution concepts to depend on the payoff parameters. 

\begin{definition}
	A \emph{solution concept} is a function $S:\Theta\to 2^H$.
\end{definition}

For example, suppose that players choose actions simultaneously and only use rationalizable strategies, i.e., strategies that survive the iterated elimination of strictly dominated strategies.
%%%
Let $S_{R}(\theta)$ be the set of $h$ such that, given the payoff indices $\pi(x,e,\theta)$, $h(x, e)$ assigns positive probability only to rationalizable outcomes for all $x$ and $e$.

The Nash hypothesis is that the behavior of the players always constitutes NE of the simultaneous-move game in pure or mixed strategies. 
%%%
There can be multiple NE and, in such cases, there is no consensus on which equilibria are more likely to arise. 
%%%
In order to assume as little as possible about the equilibrium selection, one must allow for arbitrary mixtures of equilibria. 
%%%
The actual distribution of outcomes could be any point in the convex hull of the set of the distribution over outcomes implied by NE. 
%%%
Let $S_{N}(\theta)$ be the set of $h$ such that, for all $x$ and $e$, $h(x,e)$ belongs to the convex hull of the set of distributions over outcomes implied by NE of the game given the payoff indices $\pi(x, e,\theta)$.
%%%
$S_N(\theta)$ exactly captures the predictions of the Nash hypothesis. 

Note that the NE solution concept is \emph{nested} into rationalizability.
%%%
That is, $S_{N}(\theta)\subseteq S_{R}(\theta)$ for all $\theta$.
%%%
Also $S_{N}$ is a \emph{convex} solution concept in that $S_{N}(\theta)$ is a convex set for all $\theta$.

The distribution of play completely characterizes behavior. 
%%%
However, in general, there are at least two reasons to work with restrictions that are coming from Economic Theory, that is, solution concepts.
%%%
First, solution concepts might make for more credible counterfactual analyses because they are also supported by nonempirical arguments (\citet{dawid16}).
%%%
For instance, one may argue that the policy intervention considered in Section~\ref{sec:eg-policy} would not affect whether firms play PNE.
%%%
However, since only one firm is subsidized, it is possible that the subsidized firm will be more likely to enter in markets with multiple PNE. 
%%%
In this case, the distribution of play would not be policy invariant,
but the predictions based on the solution concept would remain accurate.

A second reason to focus on solution concept is portability.
%%%
A solution concept can make predictions related to changes in some fundamental characteristics of the environment.
%%%
For example, NE is well defined for two-player and three-player games.
%%%
In contrast, the distribution of play cannot be easily extrapolated to make predictions if the number of players changes.

Moreover, a solution concept might be relevant beyond the specific application being considered. 
%%%
Many solution concepts from Economic Theory are general theories of behavior. 
%%%
Finding evidence in support for a solution concept in one setting, 
provides support for its use in other settings. 
%%%
For instance, \cite{walker01}, \cite{chiappori02}, and \cite{wooders16} have tested the implications of the Nash hypothesis in the context of penalty kicks and tennis serves.%
%%%
\footnote{%
	These papers consider only zero-sum games.
	%%%
	It is not entirely clear whether it is possible to generalize their methodology to general-sum games. 
}
%%%
And their analysis is often used to justify the use of NE in general settings unrelated to sports. 	
	
For most of the text, we assume that the players' behavior is rationalizable. 
%%%
This assumption is common in the literature (see, for instance, \citet{aradillas08} and \citet{kline15}).
%%%
Section~\ref{sec:beyond} relaxes this assumption.

\begin{assumption}[Rationalizability]\label{ass:rat}
	$h_{0} \in S_{R}(\theta_{0})$.
\end{assumption}

\subsection{Relation Between Solution Concepts, Distributions of Play, and Selection Mechanisms}

This section clarifies the relation between solution concepts, the distribution of play, and selection mechanisms.
%%%
Our distribution of play is a complete econometric model in the sense of \citet{tamer03} and \citet{manski88}
in that it ``asserts that a random variable $\rand{y}$ is a function of a random pair $(\rand{x}, $[$\rand{e}$]$)$ where $\rand{x}$ is observable and [$\rand{e}$] is not'' (\citet{tamer03}, pp.\ 150).
%%%
In other words, $h_0$ is an ``empirical'' solution concept that completely describes behavior of players without any economic restrictions.
%%%
In contrast, many solution concepts arising from Economic Theory are generally incomplete in that,
even knowing the value of the parameters and all the characteristics of the environment, there can be multiple solutions.
%%%
Both NE and rationalizability fall under this category.
%%%
Our approach is to take the distribution of play as a primitive, and 
model incomplete solution concepts as sets of complete models that depend on the parameters of the environment. 

An alternative approach to ours is to define a solution concept as a random set $\mathrm{Sol}(\rand{x},\rand{e},\theta)$ consisting of possible distributions $p$ over outcomes, which depend on the characteristic of the environment and the payoffs (e.g., \cite{BMM} and \citet{BHR}). 
%%%
Then, one would complete the model with a selection mechanism that assigns probabilities $\mathrm{sel}(\blank|x, e,\theta)$ to the different possible distributions emerging from the solution concept. 
%%%
The distribution of play could be defined as a weighted average of the different distributions with weights determined by the selection mechanism. 
%%%
For instance, if $\mathrm{Sol}(x, e,\theta)$ is finite,
\[
h(y,x,e) = \sum_{p\in \mathrm{Sol}(x, e,\theta)} p(y) \cdot \mathrm{sel}(p|x, e,\theta).
\]

Under some technical measurability assumptions, both approaches are mathematically equivalent in terms of the relation between solution concepts, distributions of play, and the data (see Section 2 in \cite{BMM}).
%%%
The difference between the two approaches is that our approach emphasizes the distribution of play, which is well defined independently of any solution concept. 
%%%
In contrast, selection mechanisms can only be defined relative to a specific solution concept. 
%%%
The following example demonstrates the relation between them.

\setcounter{aux}{\value{example}}
\setcounter{example}{\value{egentry}}
\begin{example}(continued)
	Fix some $x$ and suppose that $I=\{1,2\}$, 
	the strategic effects have negative signs ($\delta^0_{12}(w),\delta^0_{21}(w)<0$),
	and the firms always play NE, selecting each NE with equal probabilities.
	%%%
	Depending on the realizations of $\rand{e}$ there are at most three NE.
	%%%
	When there is a unique NE, then the distribution of play assigns full probability to the equilibrium outcome.
	When there are three NE (firm $1$ monopoly, firm $2$ monopoly, and a mixed one), then $h_0(x,e)$ is the equally-weighted mixture of the distributions over outcomes implied by these three NE.
	For instance, 
	\begin{align*}
	h_0((1,0),x,e)
	& =p_{1}((1,0),x,e)\cdot 1/3+p_{2}((1,0),x,e)\cdot 1/3+p_m((1,0),x,e)\cdot 1/3\\
	& = 1\cdot 1/3+0\cdot 1/3+p_m((1,0),x,e)\cdot 1/3\\
	& = \dfrac{1}{3} + \dfrac{1}{3}\cdot\dfrac{e_2-\delta^{0}_{22}(w)-\beta^{0}_2(w)z_2}{\delta^{0}_{21}(w)}
	\cdot\left(1-\dfrac{e_1-\delta^{0}_{11}(w)-\beta^{0}_1(w)z_1}{\delta^{0}_{12}(w)}\right),
	\end{align*}
	where $p_{i}(y,x,e)$ and $p_m(y,x,e)$ are the probability that outcome $y$ is played under firm $i$ monopoly NE and the mixed strategy NE, respectively.
	%%%
	Note that $h_0((1,0),x,e)$ is a function of $w$, $\beta^{0}_1(w)z_1 - e_1$, and $\beta^{0}_2(w)z_2 - e_2$.
	%%%
	Hence, the distribution of play in this example satisfies our exclusion restriction. 
\end{example}
\setcounter{example}{\value{aux}}
	
%%%%%%%%%%%%%%%%%%%%%%%%%%%%%%%%%%%%%%%%%%%%%%%%%%%%%%%%%%%%%%%%%%%%%%%
%%%%%%%%%%%%%%%%%%%%%%%%%%%%%%%%%%%%%%%%%%%%%%%%%%%%%%%%%%%%%%%%%%%%%%%
\section{Discernibility of Rationalizable Solution Concepts}
\label{sec:main}
%%%
\subsection{Definition of Discernibility}
\label{sec:main-definition}
Recall that $\theta_0$ and $h_0$ denote the true payoff parameters and the true distribution of play.
%%%
Let $\Psi\subseteq\Theta\times H$ denote a set of possible values that $(\theta_0,h_0)$ can take.
%%%
% That is, $\Psi$ is set of parameters that satisfy assumptions~\ref{ass:u}--\ref{ass:rat}.
We are interested in whether solution concepts, in particular NE, are discernible according to the following definition:
\begin{definition}\label{def:discernibility}
	Given a set $\Psi\subseteq \Theta\times H$, a solution concept $S$ is said to be \emph{discernible} relative to $\Psi$ if there do \emph{not} exist $(\theta,h)$, $(\theta',h')\in\Psi$ such that
	$h\in S(\theta)$, $h'\not\in S(\theta')$, and
	\[
	%\label{eqn:discernibility}
	\Exp{h(\rand{x},\rand{e})|\rand{x};\theta} 
	= \Exp{h'(\rand{x},\rand{e})|\rand{x};\theta'} \as.
	\]
\end{definition}
Note that 
\[
\Exp{h_0(y,\rand{x},\rand{e})|\rand{x};\theta_0} = \Pr(\rand{y}=y|\rand{x})\ \as,
\]
for all $y\in Y$.
%%%
Hence, $\Exp{h_0(\rand{x},\rand{e})|\rand{x};\theta_0}$ is identified (can be consistently estimated from observed data on outcomes and covariates).
%%%
Thus, the definition of discernibility leads to two properties that fully characterize the relationship between the observed (estimable) distribution of $\rand{y}$ conditional on $\rand{x}$ and any given solution concept $S$: 
\begin{enumerate}
\item \label{discernibility-yes} If the data is generated by a given solution concept, then it cannot be explained by something else: if $h_0\in S(\theta_0)$, then
	\[
	\left(\Pr(\rand{y}=y|\rand{x})\right)_{y\in Y}\neq\Exp{h(\rand{x},\rand{e})|\rand{x};\theta}
	\]
	with positive probability for all $(\theta,h)\in\Psi$ such that $h\not\in S(\theta)$.
\item \label{discernibility-no} If the data is not generated by a given solution concept, then it cannot be explained by it: if $h_0\not\in S(\theta_0)$, then
	\[
	\left(\Pr(\rand{y}=y|\rand{x})\right)_{y\in Y}\neq\Exp{h(\rand{x},\rand{e})|\rand{x};\theta}
	\]
	with positive probability for all $(\theta,h)\in\Psi$ such that $h\in S(\theta)$.
\end{enumerate}

%%%%%%%%%%%%%%%%%%%%%%%%%%%%%%%%%%%%%%%%%%%%%%%%%%%%%%%%%%%%%%%%%%%%%%%
%%%%%%%%%%%%%%%%%%%%%%%%%%%%%%%%%%%%%%%%%%%%%%%%%%%%%%%%%%%%%%%%%%%%%%%
\subsection{Identification of \texorpdfstring{$\theta_0$}{the Payoff Parameters} and \texorpdfstring{$h_0$}{the Distribution of Play}}
%%%
\label{sec:main-ID}
%%%
In order to establish discernibility of solution concepts, we first establish identification of the payoff parameters and the distribution of play.
%%%
Later on, we also consider environments where identification of payoff parameters fails to hold.

\begin{proposition}\label{prop:u}
	Under assumptions~\ref{ass:z}--\ref{ass:rat}, $\theta_{0}$ is identified.
\end{proposition}

The proof of the proposition is in Appendix~\ref{sec:proofs}. 
%%%
The identification of the payoff parameters $\beta_{0}$ and $\alpha_{0}$ follows standard arguments that exploit the large support of our player-specific covariates.
%%%
However, to the best of our knowledge, Proposition~\ref{prop:u} is the first result in the literature that identifies the unknown correlation structure of payoffs in games of complete information under solution concepts weaker than PNE.%
%%%
\footnote{%
	See \citet{kline15} for an identification result assuming either PNE or independent shocks.}
%%%
The proof involves looking at the limits of the partial derivative of $\Pr(\rand{y}=y|\rand{x}=x)$	
with respect to one of the excluded covariates for a specific outcome vector $y$ along specific rays in the support of excluded covariates.%
%%%
\footnote{%
    Although the rays themselves  have probability zero, the partial derivatives use information in a neighborhood of those rays.
    %%%
    Because the covariates are continuous, these open neighborhoods have positive probability and thus observable implications.
}

To identify $h_0$ note that, under Assumption~\ref{ass:ER}, the excluded covariates $\rand{z}$ generate variation in the observed conditional distribution over outcomes without changing $h_0$.
%%%
This is because $h_0$ is affected by $\rand{z}$ \emph{only} via the conditional distribution of the index $\rand{v}$ conditional on $\rand{z}$.
%%%
This exogenous variation yields the following identification result.  

\begin{proposition}\label{prop:dop}
	Under assumptions~\ref{ass:z}--\ref{ass:ER}, if $\theta_{0}$ is identified, 
	then $h_{0}$ is identified.
\end{proposition}
\begin{proof}
	Suppose towards a contradiction that there exist $h' \neq h_{0}$ such that
	\[
	\Pr(\rand{y}=y|\rand{x}=x)
	= \Exp{h_{0}(y,\rand{x},\rand{e})|\rand{x}=x;\theta_{0}}
	= \Exp{h'(y,\rand{x},\rand{e})|\rand{x}=x;\theta_{0}}.	
	\]
	for all $y\in Y$ and $x\in X$.
	%%%
	By Assumption~\ref{ass:ER} there exist $\tilde h_{0}$ and $\tilde h'$ such that 
	\[
	%\label{eqn:proof-dop-A}
	\Exp{\tilde h_{0}(y,\rand{w},\rand{v})-\tilde h'(y,\rand{w},\rand{v})|\rand{x}=x;\theta_{0}}=0,
	\]
	for all $y\in Y$ and $x\in X$, where $\rand{v}$ denotes the index from Assumption~\ref{ass:ER}.
	%%%
	The collection of the conditional distributions of $\rand{v}|\rand{x}=x$ 
	%, $\{P_{\rand{v}|\rand{x}}(\blank|x)\:\:x\in X\}$, 
	forms a complete exponential family (see Theorem 2.12 in \cite{brown86}).
	%%%
	Hence, %(\ref{eqn:proof-dop-A}) implies that 
	\[
	\tilde h_{0}(y,\rand{w},\rand{v})- \tilde{h}'(y,\rand{w},\rand{v})=0\quad \as	
	\]
	for all $y\in Y$. 
	%	Consequently. 
	%	\begin{align*}
	%		h_{0}(y,\rand{x},\rand{e})- h'(y,\rand{w},\rand{v})\ \as	
	%	\end{align*}	
	%%%
	This establishes identification of $h_{0}$.
\end{proof}
The proof of Proposition~\ref{prop:dop} uses bounded completeness of the family of normal distributions.
%%%
This property is not unique to normal distributions and is satisfied by many other parametric families.
%%%
See Appendix~\ref{sec:proofs} for more details.

%%%%%%%%%%%%%%%%%%%%%%%%%%%%%%%%%%%%%%%%%%%%%%%%%%%%%%%%%%%%%%%%%%%%%%%
%%%%%%%%%%%%%%%%%%%%%%%%%%%%%%%%%%%%%%%%%%%%%%%%%%%%%%%%%%%%%%%%%%%%%%%
\subsection{Testable Solution Concepts}
\label{sec:main-discernibility}
%%%
Note that the definition of discernibility has the form ``there do \emph{not} exist $(\theta,h)$, $(\theta',h')\in\Psi$ such that\ldots.''
%%%
Hence, the smaller $\Psi$ is (the more restrictions are imposed), the easier it is to establish discernibility of a solution concept.
%%%
In particular, if $\Psi$ is a singleton, then every solution concept is trivially discernible. 
%%%
On the other hand, if $\Psi$ is very big (very few restrictions are imposed), then one should not expect many solution concepts to be discernible.
%%%
In this section, we provide two results.
%%%
The first one establishes discernibility of solution concepts under a small number of restrictions (Theorem~\ref{thm:main}).
%%%
The second one shows absence of discernibility despite $\Psi$ being very small (Proposition~\ref{prop:nash}).
%(e.g. PNE is not discernible in our motivating example).
%%%
\begin{theorem}\label{thm:main}
	Let $S$ be any solution concept nested into $S_{R}$ and let $\Psi_{R}$ be the set of parameters that satisfy assumptions~\ref{ass:z}--\ref{ass:rat}.
	%%%
	$S$ is discernible relative to $\Psi_R$.
\end{theorem}
\begin{proof}
	Assume towards a contradiction that there exists $(\theta,h)$ and $(\theta',h')$ that satisfy assumptions~\ref{ass:z}--\ref{ass:rat}, and such that $h\in S(\theta)$, $h'\not\in S(\theta')$, and
	\[
	\Exp{h(\rand{x},\rand{e})|\rand{x};\theta}=\Exp{h'(\rand{x},\rand{e})|\rand{x};\theta'}\:\as.
	\]
	Propositions~\ref{prop:u} and~\ref{prop:dop} then imply that $h=h'$ and $\theta=\theta'$. The latter is not possible since by assumption $h'\not\in S(\theta')=S(\theta)$. The contradiction completes the proof.
\end{proof}

Under our assumptions, Theorem~\ref{thm:main} implies that we can discern {any} solution concept which implies rationalizable behavior, including $S_{N}$. 

\begin{corollary}\label{cor:nash}
$S_{N}$ is discernible relative to $\Psi_R$.
\end{corollary}

Strengthening Assumptions~\ref{ass:z}--\ref{ass:rat} would not change the conclusion of Corollary~\ref{cor:nash}, because it would only shrink the set $\Psi_R$.
%%%
In contrast, the following proposition establishes the importance of our exclusion restrictions (Assumption~\ref{ass:ER}). 

\begin{proposition}\label{prop:nash}
Suppose Assumptions~\ref{ass:z}, \ref{ass:e}, and~\ref{ass:rat} hold.
%%%
Either $S_N(\theta_0)=S_{R}(\theta_0)$ or
$S_N$ is \emph{not} discernible relative to $\Psi^* = \{\theta_0\}\times S_{R}(\theta_0)$.
\end{proposition}
It is important to note that the payoff parameters are known under $\Psi^*$.
%%%
That is, without Assumption~\ref{ass:ER}, the NE solution concept is not discernible even in settings with fully known payoff structure (e.g., laboratory experiments).

%%%%%%%%%%%%%%%%%%%%%%%%%%%%%%%%%%%%%%%%%%%%%%%%%%%%%%%%%%%%%%%%%%%%%%%%%%%%%%%%%%%%%%%%%%%%%%%
\section{Beyond Rationalizability and Identification of \texorpdfstring{$\theta_0$}{the Payoff Parameters}}
\label{sec:beyond}
%%%
The sole purpose of Assumption~\ref{ass:rat} is to establish identification of the payoff parameters.
%%%
However, point identification of $\theta_0$ is {not} necessary for discernibility of solution concepts.
%%%
In this section, we consider two departures from rationalizability that lead to partial identification of the payoff parameters.

Suppose that firms are ambiguity averse in the sense of \cite{gilboa89}.
%%%
That is, suppose that each firm ranks its actions in terms of its {minimum} possible payoff, and then chooses the action that maximizes this minimum.
%%%
This action is called the \emph{maxmin} action, and is generically unique (indifference between actions is ruled out by continuity of the distribution of the utility shocks). 
%%%
Let $S_{M}(\theta_{0})$ be the set of distributions of play that assign full probability to maxmin actions given the payoff parameter $\theta_0$. 
%%%
The middle panel of Figure~\ref{fig:beyond} illustrates the predictions of this solution concept for the two-firm entry-game from Example~\ref{eg:entry}.

Another possibility is that the firms are colluding.
%%%
Suppose that firms can compensate each other via transfers that are not observed by the researcher. 
%%%
In this case, the firms could agree to choose \emph{collusive} outcomes that maximize the sum of their individual profits, even if doing so does not maximize the individual profits of some of them. 
%%%
Let $S_{C}(\theta_{0})$ be the set of distributions of play that assign full probability to collusive outcomes. 
%%%
The predictions of the collusive solution concept are illustrated in the left panel of Figure~\ref{fig:beyond} given the parametrization from Example~\ref{eg:entry}.

\begin{figure}[t]
	\centering
	\psset{unit=4.5mm}
	\begin{pspicture}(2,10)(31,-3)
	\newrgbcolor{uno}{0.99 0.40 0.50}
	\newrgbcolor{dos}{0.4 0.5 0.7}
	\newrgbcolor{dos}{0.4 0.5 0.7}
	\scriptsize%	
	%\psgrid%
	% Guides
	\psset{hatchsep=1pt,hatchwidth=0.5pt,linestyle=none}%2,9,16,23
		\psframe[linestyle=solid](2,-2.5)(3,-1.5)
			\rput[l](3.5,-2){duopoly}
		\psframe*[linecolor=blue!10](6.5,-2.5)(7.5,-1.5)
			\rput[l](8,-2){no entry}
		\psframe[fillstyle=vlines,hatchcolor=dos](11.5,-2.5)(12.5,-1.5)
			\rput[l](13,-2){firm $1$ monopoly}
		\psframe[fillstyle=hlines,hatchcolor=uno](18.5,-2.5)(19.5,-1.5)
			\rput[l](20,-2){firm $2$ monopoly}
		\psframe[fillstyle=hlines,hatchcolor=uno,hatchcolor=red!60!blue](25.5,-2.5)(26.5,-1.5)
		\psframe[fillstyle=vlines,hatchcolor=dos,hatchcolor=red!60!blue](25.5,-2.5)(26.5,-1.5)
			\rput[l](27,-2){multiplicity}		
	\psreset%
	% NASH	\begin{pspicture}(-2.5,-2.5)(7.5,7.5)
	\rput(30.5,7.5){
		% Aereas
		\psset{hatchsep=1pt,hatchwidth=0.5pt,linestyle=none}
			\psframe*[linecolor=blue!10](0,0)(2,2)
			\psframe[fillstyle=hlines,hatchcolor=uno,hatchcolor=red!60!blue](0,0)(-2,-3)
			\psframe[fillstyle=vlines,hatchcolor=dos,hatchcolor=red!60!blue](0,0)(-2,-3)
			\pspolygon[fillstyle=hlines,hatchcolor=uno](-7,-3)(-2,-3)(-2,0)(0,0)(0,2)(-7,2)
			\pspolygon[fillstyle=vlines,hatchcolor=dos](-2,-7)(-2,-3)(0,-3)(0,0)(2,0)(2,-7)		
		\psreset%
		% Lines
			\psline(-2,-7)(-2,-3)(-7,-3)
			\psframe(0,0)(-2,-3)
			\psline(0,2)(0,0)(2,0)
		% Labels
		\psset{linestyle=none,fillstyle=solid,opacity=0.95,framearc=1}
			\rput[t](-2,-7){\psframebox{$\delta_{12}^0(w)$}}
			\rput[b]{90}(-7,-3){\psframebox{$\delta_{21}^0(w)$}}
			\rput[b](0,2){\psframebox{$0$}}
			\rput[l](2,0){\psframebox{$0$}}
	}
	%% MINIMAX \begin{pspicture}(-2.5,-2.5)(7.5,7.5)
	\rput(19.0,7.5){
		% Aereas
		\psset{hatchsep=1pt,hatchwidth=0.5pt,linestyle=none}
			\psframe*[linecolor=blue!10](-2,-3)(2,2)
			\psframe[fillstyle=hlines,hatchcolor=uno](-7,-3)(-2,2)
			\psframe[fillstyle=vlines,hatchcolor=dos](-2,-7)(2,-3)
		\psreset%
		% Lines
			\psline(-2,-7)(-2,2)
			\psline(-7,-3)(2,-3)
		% Labels
		\psset{linestyle=none,fillstyle=solid,opacity=0.95,framearc=1}
			\rput[t](-2,-7){\psframebox{$\delta_{12}^0(w)$}}
			\rput[b]{90}(-7,-3){\psframebox{$\delta_{21}^0(w)$}}
	}
	%% COLLUSION \begin{pspicture}(-7.5,-7.5)(2.5,2.5)
	\rput(7.5,7.5){
		% Aereas
		\psset{hatchsep=1pt,hatchwidth=0.5pt,linestyle=none}
			\psframe*[linecolor=blue!10](0,0)(2,2)
			\pspolygon[fillstyle=hlines,hatchcolor=uno](-7,-5)(-5,-5)(0,0)(0,2)(-7,2)
			\pspolygon[fillstyle=vlines,hatchcolor=dos](-5,-7)(-5,-5)(0,0)(2,0)(2,-7)		
		\psreset%
		% Lines
			\psline(-5,-7)(-5,-5)(-7,-5)
			\psline(0,0)(-5,-5)
			\psline(0,2)(0,0)(2,0)
		% Labels
		\psset{linestyle=none,fillstyle=solid,opacity=0.95,framearc=1}
			\rput[t](-3,-7){\psframebox{$\delta_{12}^0(w)+\delta_{21}^0(w)$}}
			\rput[b]{90}(-7,-3){\psframebox{$\delta_{12}^0(w)+\delta_{21}^0(w)$}}
			\rput[b](0,2){\psframebox{$0$}}
			\rput[l](2,0){\psframebox{$0$}}
	}
	\end{pspicture}
	\caption{Three solution concepts for two-firm entry games with
		$\delta_{12}(w),\delta_{21}(w) < 0$ and
		$\delta_{ii}^0(w) + \beta_i^0(w)z_i = 0$ for $i=1,2$:
		collusion (left), maxmin (middle), and rationalizability (right).
		}
	\label{fig:beyond}
\end{figure}
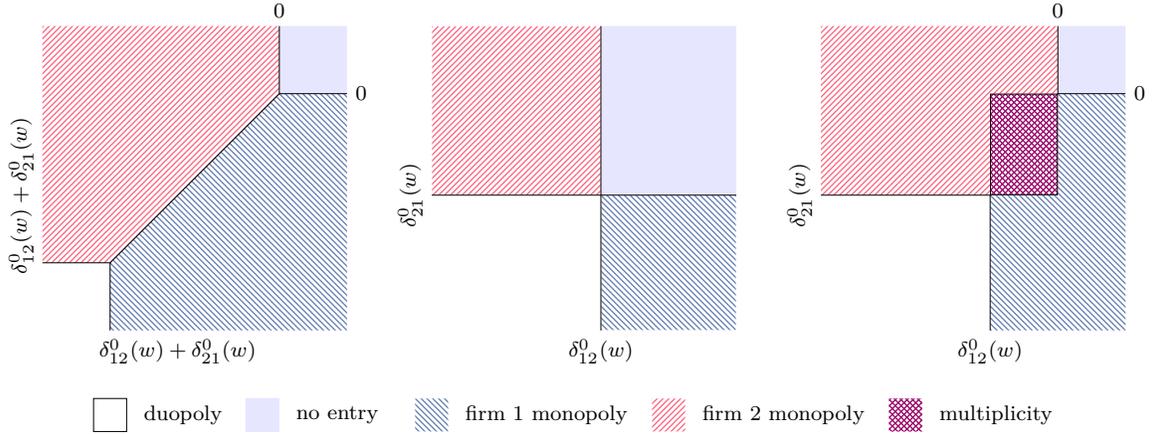

Formally, we assume that the behavior of players is consistent with either rationalizability, maxmin, or collusive behavior.
\begin{assumption}\label{ass:solutions}
	$h_0\in \bar S(\theta_0)=S_{R}(\theta_{0})\cup S_{M}(\theta_{0})\cup S_{C}(\theta_{0})$.
\end{assumption}

\begin{proposition}\label{prop:beyond}
	Suppose that assumptions~\ref{ass:z}--\ref{ass:ER}, and~\ref{ass:solutions} hold.
	Then
	\begin{enumerate}
		\item $h_0$, $\beta_{0}$, and $\Sigma^{0}$, are identified.
		\item If $h_{0}\in S_{M}(\theta_0)\cup S_{C}(\theta_0)$, then $\alpha_0$ is {not} point identified.
		\item Any solution concept nested into $\bar S$ is discernible relative to the set of parameters that satisfy assumptions~\ref{ass:z}--\ref{ass:ER}, and~\ref{ass:solutions}.
	\end{enumerate}
\end{proposition}

The proof of Proposition~\ref{prop:beyond} is in Appendix~\ref{sec:proofs}.
%%%
Under collusive behavior, we can also identify some linear combinations of $\alpha^{0}$ parameters. 
%%%
But, without imposing assumptions that would reduce the dimensionality, we cannot identify all of them. 
%%%
See Proposition~\ref{prop:minimal} in Appendix~\ref{sec:proofs} to get a sense of which linear combinations can be identified. 
%%%

We conclude this section by noting that some of out results establish discernibility of solution concepts without point identification of either $\theta_0$ or $h_0$.
%%%
In particular Proposition~\ref{prop:beyond} does not require point identification $\theta_0$,
and Proposition~\ref{prop:minimal} in the appendix establishes discernibility of $S_{R}$, $S_{M}$, and $S_C$ without point identification of $h_0$.

%%%%%%%%%%%%%%%%%%%%%%%%%%%%%%%%%%%%%%%%%%%%%%%%%%%%%%%%%%%%%%%%%%%%%%%%%%%%%%%%%%%%%%%%%%%%%%%%%%%%
%%%%%%%%%%%%%%%%%%%%%%%%%%%%%%%%%%%%%%%%%%%%%%%%%%%%%%%%%%%%%%%%%%%%%%%%%%%%%%%%%%%%%%%%%%%%%%%%%%%
\section{Conclusion}
\label{sec:conclusion}
%%%
We have defined discernibility of a given solution concept to mean that the solution concept can explain observed data if and only if the data was generated by it. 
%%%
And we have shown that any solution concept stronger than rationalizability is discernible under commonly imposed assumptions.
%%%
We have also established identification of payoff parameters, including the correlation between unobserved payoff shocks, allowing for any form of rationalizable behavior.
%%%
Our results are robust to some departures from rationalizability including ambiguity aversion and collusive behavior.
%%%
Our exclusion restriction is necessary for discernibility of the NE solution concept in some settings, even when the payoff parameters are known.

It is possible to determine whether the data can be generated by any given convex%
\footnote{A solution concept is convex if $S(\theta)$ is a convex set for all $\theta$.}
solution concept (e.g., NE).
%%%
For instance, one can construct the sets of conditional moment inequalities characterizing the solution concept (see \citet{BMM} or \citet{galichon11}).
%%%
Then, the identified set (i.e., the set of parameters that satisfy these moment inequalities) is empty if and only if the data can be generated by the solution concept.
%%%
Our results imply that one can substantially strengthen this conclusion.
%%%
The identified set of payoff parameters is empty if and only if the data is in fact generated by the solution concept.
%%%%%%%%%%%%%%%%%%%%%%%%%%%%%%%%%%%%%%%%%%%%%%%%%%%%%%%%%%%%%%%%%%%%%%%%%%%%%%%%%%%%%%%%%%%%%%%%%%%%
%%%%%%%%%%%%%%%%%%%%%%%%%%%%%%%%%%%%%%%%%%%%%%%%%%%%%%%%%%%%%%%%%%%%%%%%%%%%%%%%%%%%%%%%%%%%%%%%%%%%
\bibliography{references}

\begin{thebibliography}{}

\bibitem[\protect\astroncite{Andrews}{2011}]{andrews11}
Andrews, D. W.~K. (2011).
\newblock Examples of {L}2-complete and boundedly-complete distributions.
\newblock Discussion Paper 1801, Cowles Foundation.

\bibitem[\protect\astroncite{Aradillas-L\'opez and Tamer}{2008}]{aradillas08}
Aradillas-L\'opez, A. and Tamer, E. (2008).
\newblock The identification power of equilibrium in simple games.
\newblock {\em Journal of Business \& Economic Statistics}, 26(3):261--283.

\bibitem[\protect\astroncite{Bajari et~al.}{2010}]{BHR}
Bajari, P., Hong, H., and Ryan, S.~P. (2010).
\newblock Identification and estimation of a discrete game of complete
  information.
\newblock {\em Econometrica}, 78(5):1529--1568.

\bibitem[\protect\astroncite{Beresteanu et~al.}{2011}]{BMM}
Beresteanu, A., Molchanov, I., and Molinari, F. (2011).
\newblock Sharp identification regions in models with convex moment
  predictions.
\newblock {\em Econometrica}, 79(6):1785--1821.

\bibitem[\protect\astroncite{Bernheim}{1984}]{bernheim1984}
Bernheim, B.~D. (1984).
\newblock Rationalizable strategic behavior.
\newblock {\em Econometrica}, 52(4):1007--1028.

\bibitem[\protect\astroncite{Berry}{1992}]{berry92}
Berry, S.~T. (1992).
\newblock Estimation of a model of entry in the airline industry.
\newblock {\em Econometrica}, 60(4):889--917.

\bibitem[\protect\astroncite{Berry and Haile}{2014}]{berry14}
Berry, S.~T. and Haile, P.~A. (2014).
\newblock Identification in differentiated products markets using market level
  data.
\newblock {\em Econometrica}, 82(5):1749--1797.

\bibitem[\protect\astroncite{Bjorn and Vuong}{1984}]{bjorn84}
Bjorn, P.~A. and Vuong, Q.~H. (1984).
\newblock Simultaneous equations models for dummy endogenous variables: a game
  theoretic formulation with an application to labor force participation.
\newblock Social Science Working Paper 537, California Institute of Technology.

\bibitem[\protect\astroncite{Blundell et~al.}{2007}]{blundell07}
Blundell, R., Chen, X., and Kristensen, D. (2007).
\newblock Semi-nonparametric {IV} estimation of shape-invariant {E}ngel curves.
\newblock {\em Econometrica}, 75(6):1613--1669.

\bibitem[\protect\astroncite{Bresnahan and Reiss}{1990}]{bresnahan90}
Bresnahan, T.~F. and Reiss, P.~C. (1990).
\newblock Entry in monopoly market.
\newblock {\em The Review of Economic Studies}, 57(4):531--553.

\bibitem[\protect\astroncite{Brown}{1986}]{brown86}
Brown, L.~D. (1986).
\newblock {\em Fundamentals of statistical exponential families with
  applications in statistical decision theory}, volume~9 of {\em Lecture notes
  -- Monograph series}.
\newblock Institute of Mathematical Statistics.

\bibitem[\protect\astroncite{Chiappori et~al.}{2002}]{chiappori02}
Chiappori, P.-A., Levitt, S., and Groseclose, T. (2002).
\newblock Testing mixed-strategy equilibria when players are heterogeneous: the
  case of penalty kicks in soccer.
\newblock {\em American Economic Review}, 92(4):1138--1151.

\bibitem[\protect\astroncite{Ciliberto et~al.}{2018}]{ciliberto2018market}
Ciliberto, F., Murry, C., and Tamer, E.~T. (2018).
\newblock Market structure and competition in airline markets.
\newblock SSRN 2777820.

\bibitem[\protect\astroncite{Ciliberto and Tamer}{2009}]{ciliberto09}
Ciliberto, F. and Tamer, E. (2009).
\newblock Market structure and multiple equilibria in airline markets.
\newblock {\em Econometrica}, 77(6):1791--1828.

\bibitem[\protect\astroncite{Darolles et~al.}{2011}]{darolles11}
Darolles, S., Fan, Y., Florens, J.-P., and Renault, E. (2011).
\newblock Nonparametric instrumental regression.
\newblock {\em Econometrica}, 79(5):1541--1565.

\bibitem[\protect\astroncite{Dawid}{2016}]{dawid16}
Dawid, R. (2016).
\newblock Modelling non-empirical confirmation.
\newblock In Ippoliti, E., Sterpetti, F., and Nickles, T., editors, {\em Models
  and Inferences in Science}, volume~25 of {\em Studies in Applied Philosophy,
  Epistemology and Rational Ethics}, pages 191--205. Springer.

\bibitem[\protect\astroncite{De~Paula}{2013}]{dePaula13}
De~Paula, A. (2013).
\newblock Econometric analysis of games with multiple equilibria.
\newblock {\em Annual Review of Economics}, 5(1):107--131.

\bibitem[\protect\astroncite{Galichon and Henry}{2011}]{galichon11}
Galichon, A. and Henry, M. (2011).
\newblock Set identification in models with multiple equilibria.
\newblock {\em The Review of Economic Studies}, 78(4):1264--1298.

\bibitem[\protect\astroncite{Gauriot et~al.}{2016}]{wooders16}
Gauriot, R., Page, L., and Wooders, J. (2016).
\newblock {N}ash at {W}imbledon: evidence from half a million serves.
\newblock SSRN 2850919.

\bibitem[\protect\astroncite{Gilboa and Schmeidler}{1989}]{gilboa89}
Gilboa, I. and Schmeidler, D. (1989).
\newblock Maxmin expected utility with non-unique prior.
\newblock {\em Journal of Mathematical Economics}, 18(2):141--153.

\bibitem[\protect\astroncite{Goolsbee}{2002}]{goolsbee02}
Goolsbee, A. (2002).
\newblock Subsidies, the value of broadband, and the importance of fixed costs.
\newblock In Crandall, R. and Alleman, J.~H., editors, {\em Broadband: Should
  we Regulate High-Speed Internet Access?}, pages 278--294. Brookings
  Institution Press.

\bibitem[\protect\astroncite{Green et~al.}{2014}]{green13}
Green, E.~J., Marshall, R.~C., and Marx, L.~M. (2014).
\newblock Tacit collusion in oligopoly.
\newblock In Blair, R.~D. and Sokol, D.~D., editors, {\em Oxford Handbook of
  International Antitrust Economics}, volume~2, pages 464--497. Oxford
  University Press.

\bibitem[\protect\astroncite{Hoderlein et~al.}{2012}]{hoderlein12}
Hoderlein, S., Nesheim, L., and Simoni, A. (2012).
\newblock Semiparametric estimation of random coefficients in structural
  economic models.
\newblock Working Paper CWP09/12, CeMMAP.

\bibitem[\protect\astroncite{Jia}{2008}]{jia2008happens}
Jia, P. (2008).
\newblock What happens when {W}al-{M}art comes to town: An empirical analysis
  of the discount retailing industry.
\newblock {\em Econometrica}, 76(6):1263--1316.

\bibitem[\protect\astroncite{Jovanovic}{1989}]{jovanovic89}
Jovanovic, B. (1989).
\newblock Observable implications of models with multiple equilibria.
\newblock {\em Econometrica}, 56(6):1431--1437.

\bibitem[\protect\astroncite{Kline}{2015}]{kline15}
Kline, B. (2015).
\newblock Identification of complete information games.
\newblock {\em Journal of Econometrics}, 189(1):117--131.

\bibitem[\protect\astroncite{Manski}{1988}]{manski88}
Manski, C.~F. (1988).
\newblock {\em Analog Estimation Methods in Econometrics}.
\newblock Chapman \& Hall/CRC Monographs on Statistics \& Applied Probability.
  Chapman and Hall.

\bibitem[\protect\astroncite{Mass}{2019}]{mass19}
Mass, H. (2019).
\newblock Strategies under strategic uncertainty.
\newblock Mimeo.

\bibitem[\protect\astroncite{Morris and Shin}{2003}]{morris03}
Morris, S. and Shin, H.~S. (2003).
\newblock Global games: theory and applications.
\newblock In Dewatripont, M., Hansen, L.~P., and Turnovsky, S.~J., editors,
  {\em Advances in Economics and Econometrics: Theory and Applications, Eighth
  World Congress, Volume 1}, volume~34 of {\em Econometric Society Monographs},
  pages 56--114. Cambridge University Press.

\bibitem[\protect\astroncite{Nash}{1951}]{nash51}
Nash, J. (1951).
\newblock Non-cooperative games.
\newblock {\em Annals of Mathematics}, 52(2):286--295.

\bibitem[\protect\astroncite{Newey and Powell}{2003}]{newey03}
Newey, W.~K. and Powell, J.~L. (2003).
\newblock Instrumental variable estimation of nonparametric models.
\newblock {\em Econometrica}, 71(5):1565--1578.

\bibitem[\protect\astroncite{Pearce}{1984}]{pearce84}
Pearce, D.~G. (1984).
\newblock Rationalizable strategic behavior and the problem of perfection.
\newblock {\em Econometrica}, 52(4):1029--1050.

\bibitem[\protect\astroncite{Rubinstein}{1989}]{rubinstein89}
Rubinstein, A. (1989).
\newblock The electronic mail game: strategic behavior under almost common
  knowledge.
\newblock {\em American Economic Review}, 79(3):385--391.

\bibitem[\protect\astroncite{Tamer}{2003}]{tamer03}
Tamer, E. (2003).
\newblock Incomplete simultaneous discrete response model with multiple
  equilibria.
\newblock {\em The Review of Economic Studies}, 70(1):147--165.

\bibitem[\protect\astroncite{Walker and Wooders}{2001}]{walker01}
Walker, M. and Wooders, J. (2001).
\newblock Minimax play at {W}imbledon.
\newblock {\em American Economic Review}, 91(5):1521--1538.

\end{thebibliography}
\appendix

%%%%%%%%%%%%%%%%%%%%%%%%%%%%%%%%%%%%%%%%%%%%%%%%%%%%%%%%%%%%%%%%%%%%%%%%%%%%%%%%%%%%%%%%%%%%%%%
%%%%%%%%%%%%%%%%%%%%%%%%%%%%%%%%%%%%%%%%%%%%%%%%%%%%%%%%%%%%%%%%%%%%%%%%%%%%%%%%%%%%%%%%%%%%%%%
\section{Omitted Proofs}
\label{sec:proofs}

In this section, we first prove a version of Propositions~\ref{prop:u} for games with two players and without nonexcluded covariates $\rand{w}$, where we relax the normality and rationalizability assumptions.
%%%
This result shows that one does not need to point identify $\theta_0$ and $h_0$ in order to establish the discernibility of some solution concepts.
%%%
Then, we show how this results can be applied to games with covariates, with many players, and with many actions (Propositions~\ref{prop:u},~\ref{prop:beyond}, and~\ref{prop:multi}).

\subsection{Binary Two-Player Game}\label{sec:binary}
Suppose that $I=\{1,2\}$, $Y_i=\{0,1\}$ for $i\in I$, and 
$i$'s payoffs are given by 
\[
y_i\cdot (\alpha^{0}_{i,y_{-i}}+\beta_{i}^{0}\rand{z}_i - \rand{e}_i).
\]

\begin{assumption}\label{ass:minimal}{}\
\begin{enumerate}
	\item $\rand{z} = (\rand{z}_i)_{i=1,2}$ and $\rand{e} = (\rand{e}_i)_{i=1,2}$ are independent.
	\item $\rand{e}$ admits a probability density function (p.d.f.) $f_{\rand{e}}$ that is continuously differentiable and strictly positive on $\Real^2$.
	\item $\Exp{\rand{e}_i}=0$ and $\Exp{\rand{e}^2_i}=1$, $i=1,2$.
	%\item The support of $\rand{z}$ is $\Real^2$.
		%\item $\beta_{i}^{0}\neq 0$ for all $i\in I$.
		% 		\item $| \alpha^{0}_{1,1}-\alpha^{0}_{1,0} | \neq | \alpha^{0}_{2,1}-\alpha^{0}_{2,0} |$.
		%$\alpha^{0}_{1,1}-\alpha^{0}_{1,0}+\alpha^{0}_{2,1}-\alpha^{0}_{2,0}\neq 0$,
		%	and $\abs{\alpha^{0}_{1,1}-\alpha^{0}_{1,0}}\neq \abs{\alpha^{0}_{2,1}-\alpha^{0}_{2,0}}$.
	\end{enumerate}
\end{assumption}

Part (ii) of Assumption~\ref{ass:minimal} is a regularity condition needed for invertibility of the marginal cumulative distribution functions (c.d.f.) $F_{\rand{e}_i}$, $i=1,2$.
%%%
Part (iii) is a location and scale normalization.
%%%
Assumption~\ref{ass:minimal} is implied by Assumption~\ref{ass:e}.

\begin{assumption}\label{ass:generic}
	$\alpha^{0}_{1,1}-\alpha^{0}_{1,0} \neq \alpha^{0}_{2,1}-\alpha^{0}_{2,0}$.
\end{assumption}

Assumption~\ref{ass:generic} 
%is a generic constraint on the value of the parameters. %%%It 
requires the strategic effects of players' actions to be asymmetric.
%%%
Geometrically, it means that the multiplicity region has different
height and width. 

The next proposition shows that rationalizability, collusion, and maxmin are discernible under minimal restrictions on the distribution of shocks.
%%%
It also illustrates that point identification of neither $\theta_0$ nor $h_0$ is necessary to for solution concepts to be discernible. 

\begin{proposition}\label{prop:minimal}
	Suppose that assumptions~\ref{ass:z} and~\ref{ass:solutions}--\ref{ass:generic} hold.
	%%%
	Then
	\begin{enumerate}
		\item $\beta_i^0$ and $F_{\rand{e}_i}$ are point identified for all $i=1,2$.
		\item If $h_{0}\in S_{R}(\theta_0)$, then $\alpha_{i,y_i}^{0}$ is identified for all $i$ and $y_i$.
		\item If $h_{0}\in S_{M}(\theta_0)$, then only $\min_{y_{i}}\{\alpha^{0}_{i,y_i}\}$ is identified for all $i$.
		\item If $h_{0}\in S_{C}(\theta_0)$, then only $\alpha^{0}_{i,0}$, $i=1,2$, and $\alpha^{0}_{1,1}+\alpha^{0}_{2,1}$ are identified.
		\item $S_{R}$, $S_{C}$, and $S_{M}$ are discernible relative to the set of parameters that satisfy assumptions~\ref{ass:z} and~\ref{ass:solutions}--\ref{ass:generic}.
	\end{enumerate}
\end{proposition}
\begin{proof}
	(\emph{Step 1---Identification of $\beta_i^{0}$ and $F_{\rand{e}_i}$})
	%%%
	Define $\delta^{0}_{i}=\alpha^{0}_{i,1}-\alpha^{0}_{i,0}$, $i=1,2$.
	%%%
	Note that under rationalizability and collusive behavior
	\[
	\Pr(\rand{y}_1=0,\rand{y}_2=0|\rand{z}=z)
	= \int_{\beta^{0}_1z_1}^{\infty}\int_{\beta^{0}_2z_2}^{\infty}
	\! f_{\rand{e}}(e) \, de
	+q(z),
	\]
	where under collusive behavior
	\[
	q(z) = \begin{cases}
	0, & \delta^{0}_1+\delta^{0}_2\leq 0,\\
	-\int_{\beta^{0}_1z_1+\delta^{0}_1+\delta^{0}_2}^{\beta^{0}_1z_1}\int_{\beta^{0}_2z_2+\delta^{0}_1+\delta^{0}_2}^{\beta^{0}_2z_2}h(\beta^{0}_1z_1-e_1,\beta^{0}_2z_2-e_2) \! f_{\rand{e}}(e)\, de, & \delta^{0}_1+\delta^{0}_2>0,
	\end{cases}
	\]
	and under rationalizability
	\[
	q(z) =
	\begin{cases}
	\int_{\beta^{0}_1z_1+\delta^{0}_{1}}^{\beta^{0}_1z_1}\int_{\beta^{0}_2z_2+\delta^{0}_2}^{\beta^{0}_2z_2}\!h_0(\beta^{0}_1z_1-e_1,\beta^{0}_2z_2-e_2)f_{\rand{e}}(e)\,de,& \delta^{0}_1,\delta^{0}_2\leq 0,\\
	\int_{\beta^{0}_1z_1}^{\beta^{0}_1z_1+\delta^{0}_1}\int_{\beta_2z_2+\delta^{0}_2}^{\beta^{0}_2z_2}\!h_0(\beta^{0}_1z_1-e_1,\beta^{0}_2z_2-e_2)f_{\rand{e}}(e)\,de,& \delta^{0}_1>0,\delta^{0}_2\leq 0,\\
	\int_{\beta^{0}_1z_1+\delta^{0}_1}^{\beta^{0}_1z_1}\int_{\beta^{0}_2z_2}^{\beta^{0}_2z_2+\delta^{0}_2}\!h_0(\beta^{0}_1z_1-e_1,\beta^{0}_2z_2-e_2)f_{\rand{e}}(e)\,de,& \delta^{0}_1\leq 0,\delta^{0}_2>0,\\
	-\int_{\beta^{0}_1z_1}^{\beta^{0}_1z_1+\delta^{0}_1}\int_{\beta^{0}_2z_2}^{\beta^{0}_2z_2+\delta^{0}_2}\!h_0(\beta^{0}_1z_1-e_1,\beta^{0}_2z_2-e_2)f_{\rand{e}}(e)\,de,& \delta^{0}_1,\delta^{0}_2>0.
	\end{cases}
	\]
	Under maxmin 
	\[
	\Pr(\rand{y}_1=0,\rand{y}_2=0|\rand{z}=z) 
	= \int_{\beta^{0}_1z_1+\min\{\delta^{0}_1,0\}}^{\infty}\int_{\beta^{0}_2z_2+\min\{\delta^{0}_2,0\}}^{\infty}
	\! f_{\rand{e}}(e)\,de.
	\]
	Hence, Lemma~\ref{lemma:aux} can be applied to $p(z)=\Pr(\rand{y}_1=0,\rand{y}_2=0|\rand{z}=z)$ under all three solution concepts.
	%%%
	We can thus identify $\beta^{0}_{i}$ and the marginal c.d.f.s $F_{\rand{e}_i}$, $i=1,2$, independently of the solution concept.
	%%%
	Since $z_i$, $i=1,2$, can be rescaled, assume without loss of generality that $\beta^{0}_i=1$, $i=1,2$.	
	
	(\emph{Step 2---Identification of $\alpha_{i,y_{-i}}^{0}$})
	Next, let 
	\[
	\mu_i(z) = \Pr(\rand{y}_i=1|\rand{z}=z)
	\]
	denote the (known) probability of firm $i$ entering,
	conditional on the value of $\rand{z}$.
	%%%
	We will take limits as $\rand{z}_{-i}$ goes to $\pm\infty$.
	%and consider the limit of $\rand{y}$ as $\rand{z}_{-i}$ goes to $\pm\infty$.
	%%%
	Under any of the three solution concepts,
	\[
	\lim_{z_{-i}\to +\infty} \mu_{-i}(z ) = 1
	\quad\text{and}\quad
	\lim_{z_{-i}\to -\infty} \mu_{-i}(z) = 0.
	\]	 
	Hence, under rationalizability,
	\begin{align*}
	\lim_{z_{-i}\to-\infty} \mu_i(z)
	= F_{\rand{e}_i}(\alpha^{0}_{i,0}+z_i)
	&\implies \alpha^{0}_{i,0} = F_{\rand{e}_i}^{-1}\big(
	\lim_{z_{-i}\to-\infty} \mu_i(z)	\big)-z_i,\\
	\lim_{z_{-i}\to+\infty} \mu_i(z)
	= F_{\rand{e}_i}(\alpha^{0}_{i,1}+z_i)
	&\implies \alpha^{0}_{i,1} = F_{\rand{e}_i}^{-1}\big(
	\lim_{z_{-i}\to+\infty} \mu_i(z)	\big)-z_i,
	\end{align*}
	where we used the facts that $\lim_{z_{-i}\to+\infty} \mu_i(z)$
	and $\lim_{z_{-i}\to-\infty} \mu_i(z)$ are well defined, 
	and $F_{\rand{e}_i}$ is known (from Step 1)
	and invertible (from part~(ii) of Assumption~\ref{ass:minimal}).
	%%%
	Similarly, under collusive behavior,
	\[
	\lim_{z_{-i}\to-\infty} \mu_i(z)
	= F_{\rand{e}_i}(\alpha^{0}_{i,0}+z_i)
	\implies \alpha^{0}_{i,0} = F_{\rand{e}_i}^{-1}\big(
	\lim_{z_{-i}\to-\infty} \mu_i(z)	\big)-z_i,
	\]
	and
	\begin{align*}
	\lim_{z_{-i}\to+\infty} \mu_i(z)
	& = F_{\rand{e}_i}(\alpha^{0}_{1,1}+\alpha^{0}_{2,1}-\alpha^{0}_{-i,0}+z_i)\\
	&\implies \alpha^{0}_{1,1}+\alpha^{0}_{2,1} = F_{\rand{e}_i}^{-1}\big(
	\lim_{z_{-i}\to+\infty} \mu_i(z)	\big)-z_i+\alpha^{0}_{-i,0}.
	\end{align*}
	Finally, under maxmin, each player $i$ chooses between a payoff of $0$ and
	\[
	\min\big\{\alpha^{0}_{i,0}+\rand{z}_{i}-\rand{e}_i,\ 
	\alpha^{0}_{i,1}+\rand{z}_{i}-\rand{e}_i\big\}
	=\rand{z}_{i}-\rand{e}_i+\min\{\alpha^{0}_{i,0},\,\alpha^{0}_{i,1}\}.
	\]
	Hence, under maxmin,
	\[
	\mu_i(z)
	= F_{\rand{e}_i}\big( \min\{\alpha^{0}_{i,0},\alpha^{0}_{i,1}\} + z_i\big)
	\implies
	\min\{\alpha^{0}_{i,0},\alpha^{0}_{i,1}\} 
	= F_{\rand{e}_i}^{-1}\big( \mu_i(z)\big) - z_i.
	\]
	
	(\emph{Step 3---Discriminating Between Solution Concepts})
	%%%
	Recall that we have defined $\delta_i^{0} = \alpha_{i,1}^{0}-\alpha^{0}_{i,0}$.
	%%%
	We can discriminate between rationalizability, collusive behavior, and maxmin by 
	examining the statistics
	\[
	t_i = F_{\rand{e}_i}^{-1}\left(\lim_{z_{-i}\to+\infty}\mu_i(z)\right) - F_{\rand{e}_i}^{-1}\left(\lim_{z_{-i}\to-\infty}\mu_i (z)\right) 
%	= \left\{\begin{array}{ll}
%	\delta_1^{0} + \delta_2^{0}&\ \text{under}\ S_C\\
%	\delta_i^{0}&\ \text{under}\ S_R\\
%	0&\ \text{under}\ S_M\\
%	\end{array}\right.,
,
	\]
	$i=1,2$.
	%%%
	Under maxmin, we have $t_1 = t_2 = 0$.
	%%%
	Under collusion, we have $t_1=t_2 = \delta_1^0+\delta_2^0$.
	%%%
	Under rationalizability, we have $t_1=\delta_1^0$ and $t_2=\delta_2^0$.
	%%%
	
	Assumption~\ref{ass:generic} implies that $\delta_1^0\neq\delta_2^0$.
	%%%
	Hence, the observed data can be used to discriminate between $S_R$ and $S_C$,
	and between $S_R$ and $S_M$.
	%%%
	If $\delta_1^{0} + \delta_2^{0}\neq 0$, then $S_M$ and $S_C$ can also be discriminated. 
	%%%%
	Otherwise, $S_C(\theta)=S_M(\theta)$ for all $\theta$ and hence $S_M$ and $S_C$ are discernible. 
\end{proof}

Proposition~\ref{prop:minimal} still does not identify the correlation structure nor the distribution of play.
%%%
For that purpose, we impose the following assumptions on the distribution of the unobserved payoff shocks. 
%%%
Let $f_{\rand{e}_i}$, and $F_{\rand{e}_i|\rand{e}_j}$ denote marginal p.d.f.\ of $\rand{e}_i$ and the conditional c.d.f.\ of $\rand{e}_i$ conditional on $\rand{e}_j$, respectively.
%%%
We use $\partial_{e_i}$ to denote the partial derivatives with respect to $e_i$.

\begin{assumption}\label{ass:tail}{}\ 
	\begin{enumerate}
		\item For all $\tau\in(-1,1)$ and all $\bar{e},\ubar{e}\in\Real$ such that $\bar{e} \geq \ubar{e}$,
		there exists a constant $e^* \in\Real$ such that 
		\[
		\big[ e_1 < e^* \enspace\text{and}\enspace \tau e_1 + \ubar{e} \leq e_{2} \leq \tau e_1 + \bar{e}\big]
		\implies \partial_{e_1} f_\rand{e}(e) \geq 0.
		\]
		\item For almost all (with respect to the Lebesgue measure) $\tau \in (-1,1)$, 
		and all real numbers $\gamma_1,\bar\gamma_{2},\ubar\gamma_{1}\in \Real$
		such that $\bar\gamma_{2}\geq \ubar\gamma_{2}$, 
		\[
		\lim_{z_{1}\to-\infty}\dfrac{f_{\rand{e}_{1}}(z_{1}+\gamma_{1})}{f_{\rand{e}_{1}}(z_{1})}
		\Big[F_{\rand{e}_{2}|\rand{e}_{1}}(\tau z_{1}+\bar\gamma_{2}|z_{1}+\gamma_1)
		-F_{\rand{e}_{2}|\rand{e}_{1}}(\tau z_{1}+\ubar\gamma_{2}|z_{1}+\gamma_1)\Big]
		=0.
		\]
	\end{enumerate}  
\end{assumption}
Condition (i) in Assumption~\ref{ass:tail} requires the tail of the joint density to be convex in $e_i$ along directions $e_{-i}=\tau e_i$.
%%%
Condition (ii) controls the rates of convergence to zero of the marginal p.d.f.\ and the conditional c.d.f..
%%%
Both conditions are satisfied by distributions that have exponential tails. %%%
\begin{example}
	Assumption~\ref{ass:tail} is satisfied for the bivariate normal distribution.
	%%%
	Indeed, for the bivariate distribution with unit variances and with correlation $\rho_0$
	\[
	\partial_{e_1}f_{\rand{e}}(e_1,e_2)
	= \dfrac{\rho_0 e_2-e_1}{1-\rho_0^2}\cdot f_{\rand{e}}(e_1,e_2)
	\geq \dfrac{(\rho_0\tau-1)e_1+\min\{\rho_0\ubar{e},\rho_0\bar{e}\}}{1-\rho_0^2}\cdot f_{\rand{e}}(e_1,e_2).
	\]
	%%%
	Hence, one can take $e^*=\min\{\rho_0\ubar{e},\rho_0\bar{e}\}/(1-\tau\rho_0)$.
	%%%
	In order to verify condition (ii), note that for the normal distribution we have
	\[
	\dfrac{f_{\rand{e}_i}(z_i+\gamma_i)}{f_{\rand{e}_i}(z_i)} = 
	\exp(-\gamma_i^2/2)\cdot\exp(\gamma_iz_i),
	\]
	and 
	\begin{align*}
	F_{\rand{e}_{-i}|\rand{e}_{i}}(\tau z_{i}+\bar\gamma_{-i}&|z_{i}+\gamma_i) 
	-F_{\rand{e}_{-i}|\rand{e}_{i}}(\tau z_{i}+\ubar\gamma_{-i}|z_{i}+\gamma_i)\\
	& = \Phi\left(\dfrac{(\tau-\rho_0)e_i+ \bar\gamma_{-i}-\rho_0\gamma_i}{\sqrt{1-\rho_0^2}}\right)
	-\Phi\left(\dfrac{(\tau-\rho_0)e_i+ \ubar\gamma_{-i}-\rho_0\gamma_i}{\sqrt{1-\rho_0^2}}\right),
	\end{align*}
	where $\Phi(\blank)$ denotes the standard normal c.d.f..
	%%%
	Thus, for $\gamma_i>0$, condition~(iii) follows trivially.
	%%%
	Now, suppose $\gamma_i<0$.
	%%%
	If $\tau\neq \rho_0$, then L'H\^opital's yields:
	\begin{align*}
	& \lim_{e_1\to-\infty} \dfrac{f_{\rand{e}_i}(z_i+\gamma_i)}{f_{\rand{e}_i}(z_i)}\Big[F_{\rand{e}_{-i}|\rand{e}_{i}}(\tau z_{i}+\bar\gamma_{-i}|z_{i}+\gamma_i) 
	-F_{\rand{e}_{-i}|\rand{e}_{i}}(\tau z_{i}+\ubar\gamma_{-i}|z_{i}+\gamma_i)\Big] \\
	& =\lim_{e_1\to-\infty}\dfrac{(\tau-\rho_0) \left[\phi\left(\dfrac{(\tau-\rho_0)e_i+ \bar\gamma_{-i}-\rho_0\gamma_i}{\sqrt{1-\rho_0^2}}\right)
		-\phi\left(\dfrac{(\tau-\rho_0)e_i+ \ubar\gamma_{-i}-\rho_0\gamma_i}{\sqrt{1-\rho_0^2}}\right)\right]}%
	{-\gamma_i \exp(-\gamma_1^2/2) \cdot \exp(-\gamma_i e_i)}\\
	& =\lim_{e_1\to-\infty}\dfrac{\exp\left(-\dfrac{(\tau-\rho_0)^2}{2(1-\rho_0^2)}e_i^2\right)}{\exp(-\gamma_ie_i)},
	\end{align*}
	where $\phi(\blank)$ denotes the standard normal p.d.f..
	The latter limit equals to zero because $\tau\neq \rho_0$.
\end{example}

We also need to model the correlation between $\rand{e}_1$ and $\rand{e}_2$.
\begin{assumption}\label{ass:corr}{}\
$\rand{e}_2=\sqrt{1-\rho_0^2}\,\bm\xi+\rho_0\rand{e}_1\:\as$, where $\bm{\xi}$ is independent from $\rand{e}_1$ and $\rho_0\in(-1,1)$ is an unknown parameter.
\end{assumption}

Assumptions~\ref{ass:tail} and~\ref{ass:corr} allow us to identify the correlation between unobservables.
%%%
Note that the roles of $\rand{e}_1$ and $\rand{e}_2$ are fully interchangeable in Assumptions~\ref{ass:tail} and~\ref{ass:corr} 
(e.g., one could define $\rand{e}_1=\sqrt{1-\rho_0^2}\,\bm\xi+\rho_0\rand{e}_2\:\as$).
%%%
In order to identify $h_0$, we use the following assumption.

\begin{assumption}\label{ass:completeness}
	The family of distributions $\{F_{\rand{e}}(e - t)\:|\:t\in\Real^2\}$ is boundedly complete. That is, for any bounded function $g:\Real^2\to\Real$
	\[
	\left[ \forall t\in\Real^2,\: 
	\int_{-\infty}^\infty g(e)f_{\rand{e}}(e-t)\,de = 0\right]
	\implies g(\rand{e})=0\ \as.
	\]
\end{assumption}
Assumption~\ref{ass:completeness} is a richness condition and is satisfied by many multivariate distributions. For example, it is satisfied by the multivariate normal (see, for instance, \citet{newey03}) and the Gumbel distributions.
\begin{example}
	Assume that $\rand{e}_1$ and $\bm\xi$ are i.i.d.\ according to the Gumbell distribution with parameter $(0,1)$, 
	that is, the p.d.f.\ of $\rand{e}_{1}$ is $f(e_1) = \exp(-e_1-\exp(-e_1))$.
	%%%
	The family of distributions $\{f_{\rand{e}}(e - t)\:|\:t\in\Real^2\}$ is boundedly complete since it belongs to the two-parameter exponential family with $t$ as a parameter (\citet{brown86}). Indeed, letting $\gamma_0:=\sqrt{1-\rho_0^2}$,
	\begin{align*}
	f_{\rand e}(e-t) 
	& = f_{\rand{e}_2|\rand{e}_1}(e_2-t_2|e_1-t_1)f_{\rand e_1}(e_1 - t_1) \\
	&= \dfrac{1}{\gamma_0} \cdot f(e_1-t_1)f\left(\dfrac{e_2-t_2-\rho_0(e_1-t_1)}{\gamma_0}\right)\\
	& = G(e) \exp\left( \sum_{i=1}^{2} \eta_i(t) T_i(e) + \xi(t)\right),
	\end{align*}
	where
	\begin{gather*}
	G(e) = \dfrac{1}{\gamma_0} \cdot \exp\left( \frac{\left(\rho_0-\gamma_0\right)e_1-e_2}{\gamma_0}\right),\quad
	\xi(t) = \dfrac{1}{\gamma_0} \cdot \left(\left(\gamma_0-\rho_0\right)t_1+t_2\right),\\
	T_1(e) = \exp(-e_1), \qquad
	T_2(e) = \exp\left(-\dfrac{e_2-\rho_0e_1}{\gamma_0}\right),\\
	\eta_1(t) = \exp(t_1) \quad\text{and}\quad
	\eta_2(t) = \exp\left(\dfrac{t_2-\rho_0t_1}{\gamma_0}\right).
	\end{gather*}
\end{example}

The following proposition establishes point identification of $h_0$ and the correlation parameter $\rho_0$.
\begin{proposition}\label{prop:two-player}
	Suppose that assumptions~\ref{ass:z}, \ref{ass:ER}, \ref{ass:solutions}, \ref{ass:minimal}, and \ref{ass:tail}--\ref{ass:completeness} hold.
	Then
	\begin{enumerate}
		\item All the conclusions of Proposition~\ref{prop:minimal} hold;
		\item $\rho_0$ and $h_0$ are identified;
		\item Any solution concept nested into $\bar{S}$ is discernible relative to the set of parameter values that satisfy assumptions~\ref{ass:z}, \ref{ass:ER}, \ref{ass:solutions}, \ref{ass:minimal}, and \ref{ass:tail}--\ref{ass:completeness}.
	\end{enumerate}
\end{proposition}
\begin{proof}
Validity of the conclusions (i)--(iv) of Proposition~\ref{prop:minimal} is trivial, because they do no rely on Assumption \ref{ass:generic}, and all the other assumption from Proposition~\ref{prop:minimal} are satisfied.
%%%
Identification of $\rho_0$ follows from combining Step 1 of the proof of Proposition~\ref{prop:minimal} with Lemma~\ref{lemma:aux}.
%%%
Identification of $h_0$ and discernibility of any $S$ nested into $\bar{S}$ follows from the same arguments used in the proofs of Proposition~\ref{prop:dop} and Theorem~\ref{thm:main} in the main text.
\end{proof}

Note that, unlike Proposition~\ref{prop:minimal}, Proposition~\ref{prop:two-player} does {not} use Assumption~\ref{ass:generic} to discriminate between rationalizability, collusion, and maxmin. 
%%%
The main difference is that under the assumptions of Proposition~\ref{prop:minimal}, $h_0$ may not be point identified, and the only way to discriminate between solution concepts is to use asymmetry of the multiplicity region. 

\subsection{Proof of Proposition~\ref{prop:u}}
Fix some $w$. For notation simplicity we will drop $w$ from the notation. Since conditional on $\rand{w}=w$ the support of $\rand{z}$ is full and $\rand{z}_{i}$ enters only payoffs of player $i$, we can identify the sign of $\beta_{i}$ for every $i$ since 
\[
\lim_{z_{i}\to-\infty}\Pr\left(\rand{y}_i=0|\rand{z}=z\right)=0 \iff \beta_{i,y_i}>0
\]
for all $i$. Without loss of generality we assume that $\beta_{i}>0$ for all $i$ (if $\beta_{i}<0$ we can always use $-z_{i}$ as a covariate). Next we pick any two different players $i$ and $j$, and a profile of actions for all other players $\{y_k\}_{k\in I\setminus\{i,j\}}$. By sending $z_{k}$, $k\in I\setminus\{i,j\}$ either to $+\infty$ or to $-\infty$ we can guarantee that
\[
\lim_{z_k\to C_k,k\in I\setminus\{i,j\}}\Pr(\rand{y}_k=y_k,\:k\in I\setminus\{i,j\}|\rand{z}=z)=1,
\]
where $C_k=+\infty$ if $y_k=1$ and $C_k=-\infty$ if $y_k=0$. Thus we end up having a two player game.
%%%
Applying Proposition~\ref{prop:two-player} we identify $\beta_{i}$, $\beta_j$, $\alpha_{i,y_{-i}}$, $\alpha_{j,y_{-j}}$, and $\Sigma_{ij}$. The conclusion of the proposition follows from the fact that the choice of $w$, $i$, $j$, and $\{y_k\}_{k\in I\setminus\{i,j\}}$ was arbitrary.

%%%%%%%%%%%%%%%%%%%%%%%%%%%%%%%%%%%%%%%%%%%%%%%%%%%%%%%%%%%%%%%%%%%%%%%
%%%%%%%%%%%%%%%%%%%%%%%%%%%%%%%%%%%%%%%%%%%%%%%%%%%%%%%%%%%%%%%%%%%%%%%
\subsection{Proof of Proposition~\ref{prop:nash}}
Suppose that $S_{R}(\theta_0)\neq S_{N}(\theta_0)$.
%%%
We will construct two rationalizable distributions of play $h,h'\in S_{R}(\theta_0)$ such $h\in S_N(\theta_0)$, $h'\not\in S_N(\theta_0)$, and 
\[
%\addtag\label{eqn:proof-nash-A}
\Exp{h(\rand{x},\rand{e})|\rand{x};\theta_0} = \Exp{h'(\rand{x},\rand{e})|\rand{x};\theta_0}\ \as.
\]	 
%%%
Since $F_\rand{e}$ is absolutely continuous, players are almost surely not indifferent between outcomes. 
%%%
Hence, the fact that $S_{R}(\theta_0)\neq S_{N}(\theta_0)$ 
implies that the payoffs of at least two players must depend on their opponents' actions. 
%%%
That is, there must exist $i,j\in I$ and a profile $y_{-ij}\in \{0,1\}^{d_{I}-2}$ such that
\begin{align*}
\alpha^0_{i,(1,y_{-ij})}(\rand{w})\neq \alpha^0_{i,(0,y_{-ij})}(\rand{w})
\quad\text{and}\quad \alpha^0_{j,(1,y_{-ij})}(\rand{w})\neq \alpha^0_{j,(0,y_{-ij})}(\rand{w})  
\end{align*} 
with positive probability.

Since $\beta_0$ is known, we can make each player $k\not\in\{i,j\}$ play the action specified in $y_{-ij}$ with probability 1 by taking limits as $z_k$ goes to either $+\infty$ or $-\infty$.
%%%
Taking such limits for all $k\not\in\{i,j\}$ is {as if} players $i$ and $j$ were playing a two-player game. 
%%%
Hence, we can assume without loss of generality that $d_I=2$,
and $i=1$ and $j=2$ are the only two players. 
%%%
Moreover, for exposition purposes, we will drop $w$ from the notation. 

With these simplifications, the multiplicity region is characterized by 
\[
E(z) = \left\{
e\in \Real^{d_I}\:\big|\enspace
\min\{\alpha^0_{i,1},\alpha^0_{i,0}\} \leq e_i-\beta_i^0z_i \leq \max\{\alpha^0_{i,1},\alpha^0_{i,0}\} \enspace\text{for}\enspace i=1,2
\right\}.
\]
%%%
For $e\in E(z)$, the best response of each player depends on the action of its opponent.
%%%
For example, if $\alpha_{1,0}^0 > \alpha_{1,1}^0$,
then player $1$ prefers $y_1=1$ to $y_1=0$ if $y_2=0$,
and prefers $y_1=0$ if $y_2=1$.
%%%
Moreover, for $e\in E(z)$, the game has either zero or two PNEs,
and one mixed-strategy NE.
%%%
In the mixed-strategy NE, each firm $i$, $i=1,2$, chooses $y_i=1$ with probability
\[
\dfrac{a_{-i,0}^0 +\beta_{-i}^0 z_{-i} - e_{-i}}{a_{-i,0}^0-a_{-i,1}^0} \in (0,1).
\]
%%%
For $e\not\in E(z)$, there is a unique NE almost everywhere. 

Let $h$ be such that the players play the unique NE when $e\not\in E(z)$,
and play the {mixed-strategy} NE when $e\in E(z)$.
%%%
Let $h'$ be given by $h'(z,e) = h(z,e)$ when $e\not\in E(z)$,
and 
\[
h'(y,z,e) = \int\limits_{E(z)} \! h(y,z,\epsilon) f_{\rand{e}}(\epsilon)\,d\epsilon
\]
when $e\in E(z)$.
%%%
By construction, $h\in S_N(\theta_0)$ and 
\[
\Exp{h(\rand{x},\rand{e})|\rand{x};\theta_0} = \Exp{h'(\rand{x},\rand{e})|\rand{x};\theta_0}\ \as.
\]	 
Hence, it only remains to show that $h'$ does {not} belong to $S_{N}(\theta_0)$.

In the multiplicity region, the game has less than four PNEs.
%%%
Moreover, these PNEs are the same for all pairs $(z,e)$ such that $e\in E(z)$,
because they only depend on the sign of $\alpha_{i,0}^0-\alpha_{i,1}^0$, $i=1,2$.
%%%
Consequently, there exists an outcome $y^*$ that is {not} played in any PNE of the 
multiplicity region. 
%%%
Therefore, for $e\in E(z)$, $h(y^*,z,e)$ is the maximum probability of $y^*$
consistent with $S_N(\theta_0)$.
%%%
By construction, for every $z\in Z$, there exists a set with positive Lebesgue measure $\tilde E(z) \subseteq E(z)$ such that $h'(y^*,z,e)>h(y^*,z,e)$ for all $e\in \tilde E(z)$.
%%%
Therefore, $h'\not\in S_{N}(\theta_0)$, and we can conclude that $S_N$ is {not} discernible. 
\hfill\qed

Note that the distribution of play $h'$ in the proof of Proposition~\ref{prop:nash}
does \emph{not} satisfy our exclusion restriction.
%%%
Assumption~\ref{ass:ER} requires that, conditional on $\rand{w}$,
the joint distribution of distribution over outcomes conditional on the payoff indices
does not depend on $\rand{z}$.
%%%
In contrast, $h'$ allows this distribution to be different for every $z\in Z$.

\subsection{Proof of Proposition~\ref{prop:beyond}}
Fix any $w\in W$ and players $i,j\in I$, $i\neq j$. 
%%%
As in the proof of Proposition~\ref{prop:u}, we can turn the many-player game into a two-player game for $i$ and $j$ and any profile $\{y_k\}_{k\in I\setminus\{i,j\}}$.
%%%
Applying Proposition~\ref{prop:two-player}, we can identify $\beta^0_{i}(w)$, $\beta^0_j(w)$, and $\Sigma^0_{ij}(w)$.
%%%
Since the $w$, $i$, and $j$ are arbitrary, we can identify $\beta_0$ and $\Sigma^0$. 
%%%
Assumption~\ref{ass:ER} together with normality of $\rand{e}$ identifies $h_0$.
%%%
Hence, any $S$ nested into $\bar S$ is discernible.
%%%
The lack of identification of $\alpha_0$ under $S_M$ and $S_{C}$ follows from parts~(iii) and~(iv) of Proposition~\ref{prop:minimal}.

%%%%%%%%%%%%%%%%%%%%%%%%%%%%%%%%%%%%%%%%%%%%%%%%%%%%%%%%%%%%%%%%%%%%%%%
%%%%%%%%%%%%%%%%%%%%%%%%%%%%%%%%%%%%%%%%%%%%%%%%%%%%%%%%%%%%%%%%%%%%%%%
\section{Auxiliary Results}
\label{sec:aux}
The following lemma identifies
the marginal effect of $z_i$ (captured by $\beta_i^0$), 
the nonparametric marginal distributions of error terms ($F_{\rand{e}_i}$),
and the correlation between $\rand{e}_1$ and $\rand{e}_2$ in binary games.

\begin{lemma}
	\label{lemma:aux}
	Let $f_{\rand{e}}$ be the p.d.f.\ of $\rand{e}=(\rand{e}_1,\rand{e}_2)$
	and let $p:\Real^2 \to \Real$ be given by
	\[
	p(z)=\int_{\beta_1 z_1+\delta_1}^{\infty}\int_{\beta_2 z_2+\delta_2}^{\infty}
	\!f_{\rand{e}}(e)\,de_2\,de_1 + q(z),
	\]
	with
	\[
	|q(z)| = 
	\int_{\beta_1 z_1+\ubar\gamma_1}^{\beta_1 z_1+\bar\gamma_1}
	\int_{\beta_2 z_2+\ubar\gamma_2}^{\beta_2 z_2+\bar\gamma_2}
	\! g(\beta_1z_1-e_1,\beta_2z_2-e_2)f_{\rand{e}}(e)\,de_2\,de_1,
	\]
	for some unknown parameters $\beta_i\neq 0$, $\delta_i,\bar\gamma_i,\ubar\gamma_i\in\Real$ with $\bar\gamma_i\geq\ubar\gamma_i$, $i=1,2$, and $g:\Real^2\to[0,1]$.
	%%%
	If (i) $\Exp{\rand{e}_i}=0$ and $\Exp{\rand{e}^2_i}=1$, $i=1,2$; 
	and (ii) $f_\rand{e}$ is continuously differentiable and strictly positive on $\Real^2$;
	%(i) $\rand{z}=(\rand{z}_1,\rand{z}_2)$ is independent of $\rand{e}=(\rand{e}_1,\rand{e}_2)$;
	%(ii) $\rand{z}$ is supported on $\Real^{2}$;
	then, each of $\beta_{i}$, $\delta_i$, and the marginal c.d.f.s $F_{\rand{e}_i}$, $i=1,2$, are identified from knowing the function $p$.
	%%%
	If, moreover, 
	(iii) $\rand{e}_2=\sqrt{1-\rho^2}\bm\xi+\rho\rand{e}_1$ \as, where $\bm{\xi}$ is independent from $\rand{e}_1$ and $\rho\in(-1,1)$;
	and (iv) $f_{\rand{e}}$ satisfies the conditions from Assumption~\ref{ass:tail};
	then the correlation parameter $\rho$ is also identified from $p$.
\end{lemma}
\begin{proof}
	First note that 
	\[
	\lim_{\abs{z_i}\to+\infty} \abs{q(z)}
	\leq \lim_{\abs{z_i}\to+\infty} F_{\rand{e}_i}(\beta_i z_i+\bar\gamma_i) - F_{\rand{e}_i}(\beta_i z_i+\ubar\gamma_i)=0,
	\]
	for all $i$.
	%%%
	Hence, for all $i$,
	\[
	\lim_{z_{-i}\to +\infty}p(z)=\Char(\beta_{-i}<0)[1-F_{\rand{e}_{i}}(\beta_{i}z_i+\delta_i)]
	\]
	and
	\[
	\lim_{z_{-i}\to -\infty}p(z)=\Char(\beta_{-i}>0)[1-F_{\rand{e}_{i}}(\beta_{i}z_i+\delta_i)].
	\]
	Thus, since $f_{\rand{e}}$ is strictly positive on $\Real^2$, we can identify the sign of $\beta_i$ and $p_i(z_i)=:F_{\rand{e}_i}(\beta_i z_i+\delta_i)$, $i=1,2$.
	%%%
	Since $\beta_i\neq 0$ and we know the mean and the variance of $\rand{e}_i$, for $i=1,2$, we can also identify
	\[
	\int_{-\infty}^\infty\! t\,dp_i(t) 
	= \dfrac{1}{\beta_i}\cdot \Exp{\rand{e}_i-\delta_i}
	=-\dfrac{\delta_i}{\beta_i},
	\]
	and
	\[
	\int_{-\infty}^\infty\! t^2\, dp_i(t)
	= \dfrac{1}{\beta_i^2} \cdot \Exp{(\rand{e}_i-\delta_i)^2}
	= \dfrac{1+\delta_i^2}{\beta_i^2},
	\]
	where we used the change of variables $e_i = \beta_i t + \delta_i$.
	%%%
	As a result, since we already learned the sign of $\beta_i$,
	we can identify $\beta_i$ and $\delta_i$ from
	\[
	\beta_i^2 = 
	\left[{\int_{-\infty}^\infty\! t^2\, dp_i(t) 
		- \left(\int_{-\infty}^\infty \!t\,dp_i(t)\right)^2}\right]^{-1}
	\quad\text{and}\quad
	\delta_i = -\beta_i \int_{-\infty}^\infty\! t\,dp_i(t).
	\]
	
	If $q(z)=0$ for all $z$, 
	then we identify the joint distribution of $\rand{e}=(\rand{e}_1,\rand{e}_2)$ since $\beta_i$ and $\delta_i$, $i=1,2$, are identified and
	\[
	\Pr(\rand{e}_1\geq t_1,\rand{e}_2\geq t_2)=p\left(\dfrac{t_1-\delta_1}{\beta_1},\dfrac{t_2-\delta_2}{\beta_2}\right)
	\]
	for all $t_1,t_2\in \Real$.
	%%%
	Note that, in this case, we do not need to invoke Lemma~\ref{lemma:multiplicity}.
	%%%
	If $q(z)\neq 0$ for some $z$, we can still identify the marginal distribution of the error terms 
	given that
	\[
	F_{\rand{e}_i}(t_i) = p_i\left( \dfrac{t_i-\delta_i}{\beta_i}\right).
	\] 
	
	It only remains to identify $\rho$ in the case $q(z)\neq 0$ for some $z$.
	%%%
	We can always rescale and shift $z_i$, $i=1,2$.
	%%%
	Hence, we can assume without loss of generality that $\beta_i=1$ and $\delta_i=0$, $i=1,2$.
	%%%
	Then, it follows from Lemma~\ref{lemma:multiplicity} and the independence between $\bm\xi$ and $\rand{e}_1$
	that
	\begin{align*}
	1+\lim_{z_1\to-\infty}& \left.\dfrac{\partial_{z_1} p(z)}{f_{\rand{e}_1}(z_1)} \:\right|_{z_2=\tau z_1}
	= \lim_{z_1\to-\infty}F_{\rand{e}_2|\rand{e}_1}(\tau z_1|z_1)=\lim_{z_1\to-\infty}F_{\bm{\xi}}\left(\dfrac{(\tau-\rho)z_{1}}{\sqrt{1-\rho^2}}\right).
	\end{align*}
	for almost all $\tau\in [-1,1]$.
	%%%
	Hence, $\rho$ is the only point of jump discontinuity of the function
	$\psi:[-1,1]\to \Real$ given by
	\[
	\psi(\tau)=
	1+\lim_{z_1\to-\infty} \left. \dfrac{\partial_{z_1} p(z)}{f_{\rand{e}_1}(z_1)} \:\right|_{z_2=\tau z_1}.
	\]
	Hence, $\rho$ is also identified.
\end{proof}

%%%%%%%%%%%%%%%%%%%%%%%%%%%%%%%%%%%%%%%%%%%%%%%%%%%%%%%%%%%%%%%%%%%%%%
%%%%%%%%%%%%%%%%%%%%%%%%%%%%%%%%%%%%%%%%%%%%%%%%%%%%%%%%%%%%%%%%%%%%%%
\begin{lemma}\label{lemma:multiplicity}
	Let $\hat{p}:Z\to \Real$ be given by
	\[
	\hat{p}(z)=\int_{z_1}^{\infty}\int_{z_2}^{\infty} \!f_{\rand{e}}(e)\,de_2\,de_1 + \hat{q}(z),
	\]
	with
	\[
	\abs{\hat{q}(z)} = 
	\int_{z_1+\ubar\gamma_1}^{z_1+\bar\gamma_1}\int_{z_2+\ubar\gamma_2}^{z_2+\bar\gamma_2}
	\! g(z_1-e_1,z_2-e_2)f_{\rand{e}}(e)\,de_2\,de_1,
	\]
	where $\gamma, \tilde\gamma $ are such that $\gamma_i\leq \tilde\gamma_i$, $i=1,2$;
	$f_{\rand{e}}$ is a bivariate p.d.f.\ satisfying Assump\-tion~\ref{ass:tail}; 
	and $g:\Real^2\to[0,1]$ is an arbitrary function.
	%%%
	Then,
	\[
	1+\lim_{z_1\to -\infty} 
	\left.\dfrac{\partial_{z_1} \hat{p}(z) }{f_{\rand{e}_1}(z_1)}\:\right|_{z_2=\tau z_1}
	= \lim_{z_1\to-\infty}F_{\rand{e}_2|\rand{e}_1}(\tau z_1|z_1),
	\]
	for almost all $\tau\in(-1,1)$.
\end{lemma}
\begin{proof}
	First, using the change of variables $t_i = e_i-z_i$, $i=1,2$, we get that
	\[
	\abs{\hat{q}(z)} = \int_{\ubar\gamma_{1}}^{\bar\gamma_{1}}\int_{\ubar\gamma_{2}}^{\bar\gamma_{2}}
	\! g(-t_1,-t_2) f_{\rand{e}}(t_1+z_1,t_2+z_2) \, dt_2\, dt_1.
	\]
	Next, since $g$ takes values between zero and one and $f_{\rand{e}}$ is nonnegative, 
	\begin{align*}
	&\partial_{z_1} \abs{\hat{q}(z)}=\abs{\partial_{z_1} \hat{q}(z)}
	 = \int_{\ubar\gamma_{1}}^{\bar\gamma_{1}}\int_{\ubar\gamma_{2}}^{\bar\gamma_{2}}
	\! g(-t_1,-t_2) \partial_{z_1} f_\rand{e}(t_1+z_1,t_2+z_2)\,dt_2\,dt_1\\
	& \leq \int_{\ubar\gamma_{1}}^{\bar\gamma_{1}}\int_{\ubar\gamma_{2}}^{\bar\gamma_{2}}
	\!\abs{\partial_{z_1}f_\rand{e}(t_1+z_1,t_2+z_2)}\,dt_2\,dt_1
	 \leq \int_{z_1+\ubar\gamma_{1}}^{z_1+\bar\gamma_{1}}\int_{z_2+\ubar\gamma_{2}}^{z_2+\bar\gamma_{2}}
	\!\abs{\partial_{e_1}f_\rand{e}(e_1,e_2)}\,de_1\,de_2.
	\end{align*}
	
	Fix any $\tau\in (-1,1)$ and set $z_2=\tau z_1$.
	%%%
	For all $e_1 \in [z_1+\ubar\gamma_{1},z_1+\bar\gamma_{1}]$ we have 
	\begin{align*}
        \tau e_1-\max\left\{\tau\ubar\gamma_{1},\tau\bar\gamma_{1}\right\}\leq \tau z_1\leq \tau e_1-\min\left\{\tau\ubar\gamma_{1},\tau\bar\gamma_{1}\right\}.
	\end{align*}
	Moreover, for all $e_2 \in [z_2+\ubar\gamma_{1},z_2+\bar\gamma_{1}]$ we have
	\[
	\tau z_1+\ubar\gamma_{2}\leq e_2\leq \tau z_1+\bar\gamma_{2}.
	\]
	Combining both sets of inequalities yields 
	\[
\tau e_1-\max\{\tau\ubar\gamma_{1},\tau\bar\gamma_{1}\}+\ubar\gamma_{2}\leq e_2\leq \tau e_1-\min\{\tau\ubar\gamma_{1},\tau\bar\gamma_{1}\}+\bar\gamma_{2}.
\]
	
	Let $\ubar e=\ubar\gamma_{2}-\max\{\tau\ubar\gamma_{1},\tau\bar\gamma_{1}\}$ and $\bar e=\bar\gamma_{2}-\min\{\tau\ubar\gamma_{1},\tau\bar\gamma_{1}\}$.
	%%%
	Condition~(i) of Assumption~\ref{ass:tail} then implies that there exists some $e^*$ such that for all $e$ in the integration region, if $e_1<e^*$, then $f_\rand{e}(e)\geq 0$.
	%%%
	Therefore, there exists some $z^*$ such that for all $z_1<z^*$
	\begin{align*}
	\partial_{z_1} \abs{\hat{q}(z)}\, \big|_{z_2=\tau z_1}
	& \leq \left. \int_{\ubar\gamma_1}^{\bar\gamma_1}\int_{\ubar\gamma_2}^{\bar\gamma_2}
	\!\partial_{z_1} f_\rand{e}(t_1+z_1,t_2+z_2)\,dt_2\,dt_1 \:\right|_{z_2=\tau z_1}\\
	& = \left. \partial_{z_1} \int_{z_1+\ubar\gamma_{1}}^{z_1+\bar\gamma_{1}}
	\int_{z_2+\ubar\gamma_{2}}^{z_2+\bar\gamma_{2}}
	\! f_\rand{e}(e) \,de_2\,de_1 \:\right|_{z_2=\tau z_1}\\
	& = f_{\rand{e}_1}(z_1+\bar\gamma_1)\int_{\tau z_1+\ubar\gamma_2}^{\tau z_1+\bar\gamma_2}
	\! f_{\rand{e}_2|\rand{e}_1}(e_2|z_1+\bar\gamma_1)\, de_2 \\
	& \qquad - f_{\rand{e}_1}(z_1+\ubar\gamma_1) \int_{\tau z_1+\ubar\gamma_2}^{\tau z_1+\bar\gamma_2}
	\! f_{\rand{e}_2|\rand{e}_1}(e_2|z_1+\ubar\gamma_1)\, de_2.
	\end{align*}
	Note that, for $\gamma_1\in \{\bar\gamma_1,\ubar\gamma_1\}$ we have that
	\[
	\int_{\tau z_1+\ubar\gamma_2}^{\tau z_1+\bar\gamma_2}
	\! f_{\rand{e}_2|\rand{e}_1}(e_2|z_1+\gamma_1)\, de_2
	= F_{\rand{e}_2|\rand{e}_1}(\tau z_1+\bar\gamma_2|z_1+\gamma_1)
	-F_{\rand{e}_2|\rand{e}_1}(\tau z_1+\ubar\gamma_2|x_1+\gamma_1).
	\]
	Hence, applying condition~(ii) of Assumption~\ref{ass:tail}, it follows that
	\[
	\left. \lim_{z_1\to -\infty} 
	\dfrac{\partial_{z_1} \hat{q}(z)}{f_{\rand{e}_1}(z_1)} \,\right|_{z_2=\tau z_1} = 0,
	\]
	for almost every $\tau\in(-1,1)$.
	%%%
	Finally, the result follows from the fact that
	\[
	\partial_{z_1}\left(\int_{z_1}^{\infty}\int_{z_2}^{\infty} \! f_{\rand{e}}(e)\, de_2\,de_1 \right)
	= - f_{\rand{e}_1}(z_1) \left[1-F_{\rand{e}_2|\rand{e}_1}(z_2|z_1)\right].
	\qedhere
	\]
\end{proof}

\clearpage
%%%%%%%%%%%%%%%%%%%%%%%%%%%%%%%%%%%%%%%%%%%%%%%%%%%%%%%%%%%%%%%%%%%%%%%
%%%%%%%%%%%%%%%%%%%%%%%%%%%%%%%%%%%%%%%%%%%%%%%%%%%%%%%%%%%%%%%%%%%%%%%
\section{Online Supplementary Materials}

\subsection{Games with multiple actions}
\label{sec:app-multi}

Our results can be generalized to games with more than two actions.
%%%
Suppose that firms choose actions from $Y_i = \{1,\ldots,Y_{d_Y}\}$, $d_Y<+\infty$.
%%%
Player $i$'s payoffs from outcome $y$ are given by
\begin{align} 
\alpha_{i,y}(\rand{w}) + [\beta_{i,y_i}(\rand{w}) \rand{z}_{i,y_i} - \rand{e}_{i,y_i}].
\end{align}
Note that in contrast to the main that now we have an action-specific covariate and shock for every firm. 
\par
The following assumption is a standard location and scale normalizations of the payoffs.
\begin{assumption}\label{ass:multi} {}\
	\begin{enumerate}
		\item $\alpha_{i,(0,y_{-i})}(w)=\beta_{i,0}(w)=0$ for all $i$, $y_{-i}$, and $w$; $\rand{e}_{i,0}=0\:\as$ for all $i$.
		\item $\beta_{i,y_i}(w)\neq 0$ for all $i$, $y_i\neq 0$, and $w$.
		%\item The variance of $\rand{e}_{i,1}$ is equal to $1$ for all $i$;
	\end{enumerate}
\end{assumption}
Let $\rand{z}=(\rand{z}_{i,y_i})_{i\in I,y_i\in Y_{i}\setminus\{0\}}$ be a $d_Z=\sum_{i}d_{Y_i}$-dimensional vector of payoff relevant action-specific covariates; $x=(z\tr,w\tr)\tr$ be the vector of all observed covariates; and $\rand{e}=(\rand{e}_{i,y_i})_{i\in I,y_i\in Y_{i}\setminus\{0\}}$ be a vector of payoff shocks.
%%%
We allow shocks to be correlated and we impose no restrictions on the sign of $\alpha_{i,y}(\blank)$. 
\par
We group all payoff parameters and $\Sigma(\blank)$ into a single parameter $\theta\in\Theta$.

\begin{proposition}\label{prop:multi}
	Under assumptions~\ref{ass:z}--\ref{ass:rat}, and~\ref{ass:multi}, both $\theta_{0}$ and $h_{0}$ are identified, and any solution concept 
	$S$ nested into $S_R(\theta)$ is discernible relative to the set of parameters that satisfy assumptions~\ref{ass:z}--\ref{ass:ER}, and~\ref{ass:multi}.
\end{proposition}

\begin{proof}
	Similar to the proof of Proposition~\ref{prop:u} we can turn a game with many actions to a game with two actions by sending $z_{i,y_i}$ to $+\infty$ or $-\infty$, and then apply Proposition~\ref{prop:two-player} to identify the payoff parameters. Then similarly to the proof of Proposition~\ref{prop:dop}, identification of $h_0$ follows from completeness of the exponential family of distributions. The latter automatically implies discernibility of Nash solution concept in rationalizability.
\end{proof}

\subsection{Proof of Nondiscernibility of PNE and SAA}
\label{sec:eg-proof}
	Let $f_\rand{e}$ denote the p.d.f.\ of $\rand{e}$. 
	%%%
	Our assumptions imply that $f_{\rand{e}}(e_1,e_2) > 0$ and $f_{\rand{e}}(e_1,e_2)=f_{\rand{e}}(e_2,e_1)$ almost everywhere on $\Real^2$.
	For each possible outcome $y\in\{0,1\}^2$, let $p_{\mathrm{PNE}}(y;\eta')$ and $p_{\mathrm{SAA}}(y;\eta)$
	denote the probabilities of the outcome according to each of the two solution concepts under consideration. 
	
	Fix any parameter value $\eta\geq 0$.
	We will show that there exists some $\eta'\geq 0$ such that $p_{\mathrm{PNE}}(y;\eta')=p_{\mathrm{SAA}}(y;\eta)$ for every possible outcome $y$.
	%%%
	If $\eta=0$, then we can simply set $\eta'=0$. 
	%%%
	Hence, for the rest of the proof, we assume that $\eta>0$.
	
	On one hand, if $\eta'=\eta$, then 
	\begin{align*}
		p_\mathrm{PNE}((0,0);\eta') 
			& = \int_{\eta'}^\infty\int_{\eta'}^\infty \! f_\rand{e}(e_1,e_2)\,de_2\,de_1\\
			& < \int_{\eta}^\infty\int_{\eta}^\infty \! f_\rand{e}(e_1,e_2)\,de_2\,de_1
			 	+ \int_{0}^\eta\int_{0}^\eta \! f_\rand{e}(e_1,e_2)\,de_2\,de_1\\
			& = p_\mathrm{SAA}((0,0);\eta). 
	\end{align*}
	(See Figure \ref{fig:eg}).
	On the other hand, if $\eta' = 0$, then 
	\begin{align*}
	p_\mathrm{PNE}((0,0);\eta') 
		& = \int_{0}^\infty\int_{0}^\infty \! f_\rand{e}(e_1,e_2)\,de_2\,de_1\\
		& > \int_{\eta}^\infty\int_{\eta}^\infty \! f_\rand{e}(e_1,e_2)\,de_2\,de_1
			+ \int_{0}^\eta\int_{0}^\eta \! f_\rand{e}(e_1,e_2)\,de_2\,de_1\\
		& = p_\mathrm{SAA}((0,0);\eta). 
	\end{align*}
	Since $p_\mathrm{PNE}((0,0);\eta')$ is continuous in $\eta'$, there exists some 
	$\eta'\in(0,\eta)$ such that $p_\mathrm{PNE}((0,0);\eta') = p_\mathrm{SAA}((0,0);\eta)$. 
	Fix such $\eta'$.
	
	Since $p_\mathrm{PNE}((1,1);\eta') = p_\mathrm{SAA}((1,1);\eta)$ and there are only four possible outcomes
	it follows that
	\[
	p_\mathrm{PNE}((1,0);\eta')+p_\mathrm{PNE}((0,1);\eta') 
	= p_\mathrm{SAA}((1,0);\eta)+p_\mathrm{SAA}((0,1);\eta).
	\]

	Now, we will show that $p_\mathrm{PNE}((1,0);\eta')=p_\mathrm{PNE}((0,1);\eta') $
	and $p_\mathrm{SAA}((1,0);\eta)=p_\mathrm{SAA}((0,1);\eta)$.
	%%%
	This implies that the probabilities of all outcomes are the same under both solution concepts.
	%%%	
	For PNE, we have that 
	\begin{align*}
		p_\mathrm{PNE}((1,0);\eta')
			& = \int_{-\infty}^0\int_0^\infty \! f_\rand{e}(e_1,e_2)\,de_2\,de_1 
			 	+ \int_{0}^{\eta'}\int_{e_1}^{\eta'} \! f_\rand{e}(e_1,e_2)\,de_2\,de_1 \\
			& = \int_{-\infty}^0\int_0^\infty \! f_\rand{e}(e_2,e_1)\,de_1\,de_2 
				+ \int_{0}^{\eta'}\int_{e_2}^{\eta'} \! f_\rand{e}(e_2,e_1)\,de_1\,de_2\\
			& = \int_{-\infty}^0\int_0^\infty \! f_\rand{e}(e_1,e_2)\,de_1\,de_2 
				+ \int_{0}^{\eta'}\int_{e_2}^{\eta'} \! f_\rand{e}(e_1,e_2)\,de_1\,de_2\\
			& = p_\mathrm{PNE}((0,1);\eta'),
	\end{align*} 
	where the second equality follows from using the change of variables $(e_1,e_2)\to(e_2,e_1)$,
	and the third one from the symmetry fo $f_\rand{e}$.
	%%%
	The argument for SAA is completely analogous. \hfill\qed

\subsection{Alternative Entry Subsidy for the Motivating Example}\label{sec:subsidy}

A form of subsidy that is more common in practice consists of 
giving a lump sum subsidy $\hat\tau > 0$ to {any} firm that enters a market with some observable characteristics
(see, e.g., \cite{goolsbee02}). 
%%%
Under the PNE assumption,
every market that would be served without the policy would also be served with the policy. 
%%%
Hence, the policy has an unambiguously positive effect (abstracting from the cost). 
%%%
However, this need not be the case under SAA. 

\begin{proposition}
	Suppose that firms profits are given by 
	\[
	\bm\pi_i(y) = y_i\cdot\big[ \alpha + \eta(1-y_{-i}) - \rand{e}_i \big],
	\]
	firms make entry decisions in accordance with the SAA model, 
	and $\rand{e}$ is
	normally distributed with zero mean and the identity matrix as a covariance matrix.
	%	%%%
	There exists an open set $\Xi \subseteq \Real^2$ and a threshold $\bar\tau$ such that 
	if $(\alpha,\eta)\in \Xi$ and $\hat\tau < \bar\tau$,
	then the probability that a market is not served is increasing in the size of the 
	subsidy. 
\end{proposition}

\begin{proof}
	Under strategic ambiguity
	there is no entry if either $\rand{e}_i > \alpha +\eta+\hat\tau$ for $i=1,2$,
	or $\alpha + \tau < \rand{e}_i < \alpha +\eta+\hat\tau $ for $i=1,2$ 
	(See Figure~\ref{fig:eg}).
	%%%
	Hence, the probability that a market is not served as a function of $\hat\tau$ is given by 
	\begin{align}
	P(\hat\tau) & = \big[ 1 - \Phi(\alpha + \eta + \hat\tau)\big]^2
	+ \big[\Phi(\alpha + \eta + \hat\tau) - \Phi(\alpha + \hat\tau)\big]^2.
	\end{align}
	Taking derivatives
	\begin{align}
	P'(\hat\tau) & = -2 \phi(\alpha + \eta + \hat\tau)\big[ 1 - \Phi(\alpha + \eta + \hat\tau)\big] \ldots\nonumber\\
	& \ldots\ +2 \big[\phi(\alpha + \eta + \hat\tau) - \phi(\alpha + \hat\tau)\big]
	\cdot \big[\Phi(\alpha + \eta + \hat\tau) - \Phi(\alpha + \hat\tau)\big].
	\end{align}
	Evaluating when $\hat\tau=0$, $\alpha < 0$, and $\eta=-\alpha + \sqrt{-\alpha}$ yields
	\begin{align}
	\dfrac{P'(0)}{2 \phi(-\alpha)} & = %2 \phi(-z) \cdot \left(
	\left[\Phi\left(\sqrt{-\alpha}\right)-1\right]
	+ \left[1 - \dfrac{\phi(\alpha)}{\phi\left(\sqrt{-\alpha}\right)}\right]
	\cdot\left[ \Phi\left(\sqrt{-\alpha}\right) - \Phi(\alpha) \right]
	%\right)
	\end{align}
	When $\alpha\Conv -\infty$, the right-hand side converges to $1$.
	%%%
	Hence, we must have $P'(0)>0$ when $-\alpha$ is sufficiently large.
	%%%
	Since $P$ is continuous, this must also be true in an open set. 
\end{proof}
\end{document}